\documentclass[conference]{IEEEtran}
\IEEEoverridecommandlockouts
\usepackage{cite}
\usepackage{amsmath,amssymb,amsfonts}
\usepackage{algorithmic}
\usepackage[ruled,vlined]{algorithm2e}
\usepackage{graphicx}
\usepackage{textcomp}
\usepackage{subcaption}
\usepackage{url}
\usepackage{xcolor}
\usepackage[normalem]{ulem}
\usepackage{bbm}
\usepackage{amsthm}
\usepackage{enumitem}
\usepackage{tikz}

\def\BibTeX{{\rm B\kern-.05em{\sc i\kern-.025em b}\kern-.08em
    T\kern-.1667em\lower.7ex\hbox{E}\kern-.125emX}}
    
\newtheorem{theorem}{Theorem}
\newtheorem{proposition}{Proposition}
\newtheorem{lemma}{Lemma}

\newtheorem{definition}{Definition}
\allowdisplaybreaks

\begin{document}

\title{Causal Link Discovery with Unequal Edge Error Tolerance  \vspace*{-0.2in}\\
\thanks{This work has been funded in part by one or more of the following grants: ARO W911NF1910269,
ARO W911NF2410094,
DOE DE-SC0021417,
Swedish Research Council 2018-04359,
NSF CCF-2008927,
NSF RINGS-2148313,
NSF CCF-2200221,
NSF CCF-2311653,
ONR 503400-78050,
ONR N00014-22-1-2363,
NSF A22-2666-S003, and the
NSF Center for Pandemic Insights DBI-2412522
}
\thanks{Portions of this paper were presented at the 2024 International Symposium on Information Theory (ISIT) and the 2024 Asilomar Conference on Signals, Systems, and Computers.}
}
\author{\IEEEauthorblockN{Joni Shaska}
\IEEEauthorblockA{
\textit{University of Southern California}\\
shaska@usc.edu}
\and
\IEEEauthorblockN{Urbashi Mitra}
\IEEEauthorblockA{
\textit{University of Southern California}\\
ubli@usc.edu }
}
\maketitle
\thispagestyle{plain}
\pagestyle{plain}
\vspace*{-0.5in}

\begin{abstract} This paper proposes a novel framework for causal discovery with asymmetric error control, called \textit{Neyman-Pearson causal discovery}. Despite the importance of applications where different types of edge errors may have different importance, current state-of-the-art causal discovery algorithms do not differentiate between the types of edge errors, nor provide any finite-sample guarantees on the edge errors. Hence, this framework seeks to minimize one type of error while keeping the other below a user-specified tolerance level. Using techniques from information theory, fundamental performance limits are found, characterized by the R\'enyi divergence, for Neyman-Pearson causal discovery. Furthermore, a causal discovery algorithm is introduced for the case of linear additive Gaussian noise models, called $\epsilon-CUT$, that provides finite-sample guarantees on the false positive rate, while staying competitive with state-of-the-art methods.
\end{abstract}


\section{Introduction}
Understanding the underlying causes of phenomena affected by multiple variables can often be done via the representation of causal graphs \cite{pearl2009causality}.  These graphs are often assumed to be directed acyclic graphs.  Applying causal graph discovery can have utility in disciplines as diverse as topology inference in wireless networks \cite{testi2020blind}, gene networks in biology  (e.g. \cite{THOR}), the impact of medications, and optimizing the effect of advertising \cite{VANDENBROECK2018470}. As a result, there is a wide range of research regarding causal graph discovery, such as constraint-based methods like  the PC and fast-PC algorithms \cite{spirtes2000causation,sprites_book_2017}, score-based methods \cite{NOTEARS,DAGMA,DAGMA-DCE,equivariance,lasso,Golem}, and discovery methods for time-series data and state-spaces \cite{crossmaps_causal,group_sparsity_causal,sparse_bayesian_timeseries,state-space_causal}. However, despite modern applications in which one type of error is less favorable, most algorithms, with only a few notable exceptions \cite{lasso,group_sparsity_causal}, do not distinguish between different types of edge errors. Furthermore, most recoverability results in the literature are asymptotic in the number of samples and do not provide any guarantees on error rates in a finite sample setting. 

Hence, the purpose of the paper is two-fold. First, we introduce \textit{Neyman-Pearson causal discovery}, which seeks to minimize one type of edge error while keeping the other below a fixed user-specified threshold. Second, we develop an algorithm for the case of linear additive Gaussian noise models called $\epsilon-CUT$, which stands for \textit{\underline{C}ausal discovery with \underline{U}nequal edge error \underline{T}olerance}, which keeps the false positive rate below a user-specified tolerance level for any number of samples.
\par
Although work on confidence regarding causal estimates has recently started to emerge, much of this work either focuses on confidence regarding causal effects (rather than direct edges) between nodes or focuses on time-series data. For instance, \cite{schaar_information_rates} derives performance limits for causal effect estimation on individuals, rather than populations and \cite{StriederDrtonconfidence} constructs confidence intervals for causal effects in the linear additive Gaussian setting. For time-series data, \cite{HayekLower} derives converse bounds for the achievable false positive and false negative rates, and \cite{group_sparsity_causal} derives consistency conditions based on a false connection score, similar to our false positives rate.
\par
In contrast, our work deals specifically with the error rates for individual edges between nodes. In particular, we formulate an optimization problem that seeks to minimize the false negative rate of the declared edges, while keeping the false positive rate below a user-specified threshold. We derive fundamental performance limits for our framework and show that our proposed algorithm, $\epsilon-CUT$,satisfies the false-positive constraint for any number of samples.
\par
To provide finite-sample guarantees on the false positive rate, we consider the setting of linear structure equation models (LSEMs) with additive Gaussian noise, which has been extensively studied in the literature \cite{equivariance,schaar_unobserved_linear, StriederDrtonconfidence, Chen_2019_ordering}. Prior algorithms in this setting, such as AReCI \cite{residuals_latent} and ReCIT \cite{Zhang_Zhou_Guan_residuals}, rely on \textit{the principle of independent mechanisms} \cite{sprites_book_2017}, which assumes that the causes and mechanisms producing the effect (i.e., the specific functional form) are independent. Then, \cite{residuals_latent} and \cite{Zhang_Zhou_Guan_residuals}  provide independence tests by comparing estimation residuals.
However, these algorithms only provide asymptotic guarantees on recoverability, which do not hold in our finite-sample setting. Our algorithm also considers the difference between residuals but differs in two important ways. The first key observation is the variance (or energy) of a child node is always higher than that of its parent, and the second is that the residuals follow a chi-squared distribution, enabling us to derive straightforward threshold tests that facilitate finite sample analysis and yield performance guarantees.

Our algorithm also differs from traditional methods typically using a sparsity constraint or penalty. For example, \cite{equivariance} uses the $\ell_0$ penalty for linear additive Gaussian models, \cite{sparse_bayesian_timeseries} enforces sparsity through a Bayesian prior based on Gaussian processes in a time-series setting, and \cite{group_sparsity_causal} uses a combination of $\ell_1$ and $\ell_2$ regularizers for inference in multivariate autoregression models. Continuous optimization formulations, such as NOTEARS \cite{NOTEARS}, DAGMA \cite{DAGMA}, and DAGMA-DCE \cite{DAGMA-DCE}, will typically invoke sparsity through either the $\ell_1$ or $\ell_2$ regularizers. Hence, the methods often require a hyper-parameter that must be selected, and other than a few notable exceptions such as \cite{lasso}, there is typically no understanding of how different hyper-parameter values will affect the error rate \textit{a priori}. In contrast, our work requires no tuning of hyper-parameters. Once the user has specified a tolerance level, all parameters in our algorithm are defined, along with \textit{a priori} guarantees on the false positive rate for any number of given samples.

Our contributions are as follows:
\begin{enumerate}
    \item We pose edge detection in causal discovery as a collection of
     Neyman-Pearson problems. An aggregated false negative rate is minimized subject to a pre-specified tolerance on the aggregated false positive rate. Aggregation is over all edges in the causal graph.
    \item We derive the optimal detector for our framework and use it to derive upper and lower bounds on the achievable false negative rate for a given tolerance on the false positive rate (the optimal detector and bounds were given in \cite{ISITPaper} but without proof). Hence, we derive fundamental performance limits for our framework characterized by the R\'enyi divergence.
    \item We derive an algorithm for the setting of linear additive Gaussian noise models that seeks to minimize the false negative rate while keeping the false positive rate below a pre-specified tolerance level $\epsilon$. We call this method $\epsilon$-CUT.
     \item We show $\epsilon$-CUT's false positive rate will always satisfy the user-specified tolerance (the proof is absent in \cite{Asilomar}). This result requires no asymptotic assumptions on the convergence of distributions or consistency conditions and hence, is a finite-sample result. We investigate the performance of $\epsilon$-CUT compared to other causal discovery methods through a series of numerical experiments.
\end{enumerate}
The rest of this paper is organized as follows. Section \ref{Sec::problemformulation} introduces the setup and problem formulation, Section \ref{Sec::opt_limits} presents the optimal detector and fundamental performance limits for Neyman-Pearson causal discovery. Section \ref{Sec::CUT} introduces $\epsilon-CUT$. Section \ref{Sec::prioralgorithms} briefly summarized state-of-the-art methods that are used as benchmarks for the numerical studies. Section \ref{Sec::numerical} presents numerical results, and Section \ref{Sec::concl} concludes the paper. Much of our notation mimics that used in \cite{HayekLower,ISITPaper} and \cite{suncausalentropy}.

\section{Problem Formulation}\label{Sec::problemformulation}
\begin{figure}[t]
    \centering
    \vspace*{-0.2in}
    \includegraphics[width = 2.5in]{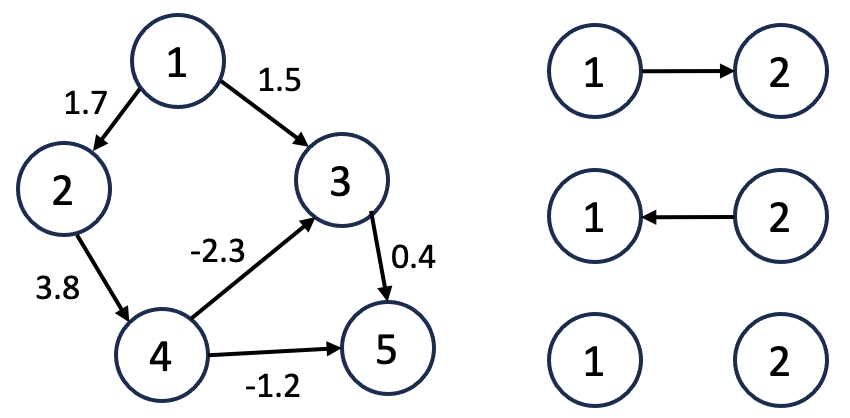}
    \caption{(L) An example, five-node, directed, acyclic graph with edge weights.  (R) Two-node graphs corresponding to the matrices $\boldsymbol{A}_0$, $\boldsymbol{A}_1$, and $\boldsymbol{A}_2$, respectively.}
    \label{fig:graphs}
\end{figure}
Consider a \textit{directed acyclic graph} (DAG) $\mathcal{G}$ with edge set $\mathcal{E}$, vertex set $\mathcal{V}$ with $\left|\mathcal{V}\right|=d$, and corresponding weighted adjacency matrix $\boldsymbol{A} \in \mathbb{R}^{d \times d}$. We assume a prior on the adjacency matrix, denoted by $\pi_{\boldsymbol{A}}$. We do this for two reasons. First, the prior will help facilitate analysis by allowing us to define key metrics. Second, having a prior on the adjacency matrix can help model prior knowledge from the application domain. We have a series of $n$ observations $\{X_k\}_{k=1}^n$ generated by the underlying distribution $P_{\boldsymbol{A}}$ conditioned on the realization of $\boldsymbol{A}$. Given a vertex $i\in\mathcal{V}$, let $X_k(i)$ be the $k$-th measurement $k=1,..,n$ of vertex $i$. Moreover, let $\mathcal{Z}(i)$ denote the parent set of vertex $i$. Assuming $\mathcal{Z}(i)$ is non-empty, let $i_j$ be the $j$-th parent of vertex $i$, $j=1,...,|\mathcal{Z}(i)|$. Define the vector
\begin{equation}
    Z_k(i) = [X_k(i_1),X_k(i_2),...,X_k(i_{|\mathcal{Z}(i)|})]^\top,
\end{equation}
i.e., the vector comprised of the $k$-th measurements of the parents of $i$. Moreover, we assume a linear structural equation model (LSEM), i.e.,
\begin{eqnarray}
    X_k & = & \boldsymbol{A}X_k + W_k,
    \label{eq:linmodel}
\end{eqnarray}
where the exogeneous input terms $\{W_k\}_{k=1}^n$ are zero-mean indenpendent Gaussian random variables with equal covariance $\sigma^2$ (which we assume is known). This assumption is commonly made in the literature \cite{equivariance,StriederDrtonconfidence,lasso,HayekLower}, and it is also worth underscoring that if one removes either the equivariance or the independence assumption, recoverability of $\boldsymbol{A}$ is no longer guaranteed even in the asymptotic case \cite{equivariance}.
 Given the underlying adjacency matrix $\boldsymbol{A}$, the goal is to recover the \textit{support} $\boldsymbol{\chi}$ of $\boldsymbol{A}$, where $\chi_{i,j}=1 \iff \boldsymbol{A}_{i,j}\neq 0$ \textit{or} $\boldsymbol{A}_{j,i}\neq 0$.
\subsection{Error Metrics}
We define the following error rates \footnote{It is important to highlight that these are in general \textit{not the probabilities of error}, and rather an averaging over errors of individual edges. To see this note that in equation \eqref{equivalenterrors} we can write out the rates in terms of individual edge error probabilities, but that the individual edge error events \textit{are not disjoint}, since we could simultaneously make errors on different edges.} which were first introduced in \cite{suncausalentropy} and further studied in \cite{HayekLower}, \cite{ISITPaper},
\begin{align}
\! \! \! \! \mbox{\textit{false positive rate}}:    \epsilon^+ &= \frac{\mathbb{E}\sum_{i,j}\mathbbm{1}\{\hat{\boldsymbol{\chi}}_{i,j}=1, \boldsymbol{\chi}_{i,j}=0\}}{\mathbb{E}\sum_{i,j}\mathbbm{1}\{\boldsymbol{\chi}_{i,j}=0\}} \label{e+} ,
    \\
\! \! \! \!  \mbox{\textit{false negative rate}}:     \epsilon^- &= \frac{\mathbb{E}\sum_{i,j}\mathbbm{1}\{\hat{\boldsymbol{\chi}}_{i,j}=0, \boldsymbol{\chi}_{i,j}\neq 0\}}{\mathbb{E}\sum_{i,j}\mathbbm{1}\{\boldsymbol{\chi}_{i,j}\neq 0\}}\label{e-} ,
\end{align}
 where the expectations are taken with respect to detector $\hat{\boldsymbol{\chi}}$ as well as the prior distribution on the matrix $\boldsymbol{A}$, $\pi_{\boldsymbol{A}}$. 
 We wish to find the detector $\hat{\boldsymbol{\chi}}$ that solves the following optimization problem.
\begin{equation}\label{NPCI}
\begin{aligned}
\inf_{\hat{\boldsymbol{\chi}}} \epsilon^- \; \; \; \; 
\textrm{s.t.} \;\; \epsilon^+ \leq \epsilon,
\end{aligned}
\end{equation}
where $0 < \epsilon < 1$. 
We underscore that the detector $\hat{\chi}$ captures the detection of {\bf all} edges in the causal graph, and that $\epsilon^+$ and $\epsilon^-$ are \textit{rates, not probabilities}; we shall provide a strategy for the detection of each individual edge, but our performance analysis will be over all edges.
\subsection{Definitions} 
Throughout this work, we will consider a binary hypothesis test for each individual edge. In particular, for an observed data set $\{X_k\}_1^n$, define the hypothesis testing problem
\begin{align}
    H_0: \boldsymbol{\chi}_{i,j} = 0, \qquad \{X_k\}_1^n \sim P_{i,j}\label{HypoTest0}\\
    H_1: \boldsymbol{\chi}_{i,j} = 1, \qquad \{X_k\}_1^n \sim Q_{i,j}\label{HypoTest1}
\end{align}
where $P_{i,j}$ denotes the density of $\{X_k\}_1^n$ conditioned on $\boldsymbol{\chi}_{i,j} = 0$ and $Q_{i,j}$ denotes the density of $\{X_k\}_1^n$ conditioned on $\boldsymbol{\chi}_{i,j} = 1$. Thus,  our problem is actually mixed in nature, since we assume a prior distribution $\pi_{\boldsymbol{A}}$ on the graph matrix $\boldsymbol{A}$; however we wish to control the aggregated error rate on the aggregated detector $\hat{\chi}$ as in a classical Neyman-Pearson-type detection problem.  
\begin{definition}
Assuming $\boldsymbol{A}\sim \pi_{\boldsymbol{A}}$, we define the following probabilities\label{errorprobs} and weights:
\begin{align}
    P^+_{i,j} &= \mathbb{P}(\hat{\boldsymbol{\chi}}_{i,j}=1|\boldsymbol{\chi}_{i,j}=0), \label{eq:P_plus}\\
    Q^-_{i,j} &= \mathbb{P}(\hat{\boldsymbol{\chi}}_{i,j}=0|\boldsymbol{\chi}_{i,j}=1).\\
    w^+_{i,j} &= \frac{\mathbb{P}(\boldsymbol{\chi}_{i,j}=0)}{\sum_{k,l}\mathbb{P}(\boldsymbol{\chi}_{k,l}=0)}, \\
    w^-_{i,j} &= \frac{\mathbb{P}(\boldsymbol{\chi}_{i,j}\neq0)}{\sum_{k,l}\mathbb{P}(\boldsymbol{\chi}_{k,l}\neq0)}. \label{eq:w_minus}
\end{align}
\hfill $\square$
\end{definition}
With Definition \ref{errorprobs}, we can substitute \eqref{eq:P_plus} - \eqref{eq:w_minus} into \eqref{e+} and \eqref{e-} and straightforwardly write out our error rates as,
    \begin{align}\label{equivalenterrors}
        \epsilon^+ = \sum_{i,j} w^+_{i,j}P_{i,j}^+, \qquad \epsilon^- = \sum_{i,j} w^-_{i,j}Q_{i,j}^-.
    \end{align}

An interesting observation is that for any edge $(i,j)$
\begin{align}
    \begin{split}
        &\mathbb{P}(\hat{\boldsymbol{\chi}}_{i,j}=1 \cap \boldsymbol{\chi}_{i,j} = 0)
        \\
        &\qquad \qquad \text{(Probability of false positive on edge } (i,j))
    \end{split}
     \\
     \begin{split}
         &\leq \mathbb{P}\bigg(\bigcup_{i\neq j}  \hat{\boldsymbol{\chi}}_{i,j}=1 \cap \boldsymbol{\chi}_{i,j} = 0 \bigg) \\
         &\ \qquad \qquad \text{(Probability of false positive on any edge)}
     \end{split}
    \\
    &\leq \sum_{i\neq j} \mathbb{P} ( \hat{\boldsymbol{\chi}}_{i,j}=1 | \boldsymbol{\chi}_{i,j} = 0 )\mathbb{P}(\boldsymbol{\chi}_{i,j} = 0) \\
    \begin{split}
        &= \bigg(\sum_{i\neq j}\mathbb{P}(\boldsymbol{\chi}_{i,j}= 0) \bigg) \sum_{i\neq j} \mathbb{P} ( \hat{\boldsymbol{\chi}}_{i,j}=1 | \boldsymbol{\chi}_{i,j} = 0 )
        \\
        & \qquad \qquad \qquad \qquad \qquad \qquad \frac{\mathbb{P}(\boldsymbol{\chi}_{i,j} = 0)}{\bigg(\sum_{i\neq j}\mathbb{P}(\boldsymbol{\chi}_{i,j} = 0) \bigg)}
    \end{split}
    \\
    &= \bigg(\sum_{i\neq j}\mathbb{P}(\boldsymbol{\chi}_{i,j}= 0) \bigg) \sum_{i,j} w^+_{i,j}P_{i,j}^+ = \epsilon^+\sum_{i\neq j}\mathbb{P}(\boldsymbol{\chi}_{i,j}= 0).
\end{align}
Since $\sum_{i\neq j}\mathbb{P}(A_{i,j}= 0)$ is a constant (for a fixed graph prior), we see that our false positive rate is an upper bound (up to a constant) for both the false positive probability on a given edge and the false positive probability on any edge. Hence, controlling the false positive rate also controls the false positive probability. Moreover, the false positive rate is often more computationally tractable than the probability of a false positive on any edge, so one may think of the false positive rate as a computationally tractable relaxation of the false positive probability.

\section{Fundamental Limits and Optimal Detector}\label{Sec::opt_limits}
Herein We derive the optimal detector for Problem \eqref{NPCI} which consists of a series of likelihood ratio tests between the distributions $P_{i,j}$ and $Q_{i,j}$ for each pair of vertices $i,j$. We derive upper and lower performance bounds for the optimal detector and thus the Neyman-Pearson causal discovery problem of \eqref{NPCI}; both bounds will described by the R\'enyi divergence.
\par
\begin{theorem}\label{optimalrule}
Assume $\boldsymbol{A}\sim \pi_{\boldsymbol{A}}$. Let $U$ be a uniform random variable on $[0,1]$. Then, the optimal detector $\hat{\boldsymbol{\chi}}^*$ for Problem \ref{NPCI} under the prior $\pi_{\boldsymbol{A}}$ is given as follows, where the per edge detector is given by,
\begin{equation}
    \hat{\boldsymbol{\chi}}^*_{i,j} = 
    \begin{cases}
    0, & \frac{dP_{i,j}}{dQ_{i,j}} > \frac{w_{i,j}^-}{w_{i,j}^+}\gamma \\
    0, & \frac{dP_{i,j}}{dQ_{i,j}} =  \frac{w_{i,j}^-}{w_{i,j}^+}\gamma \text{ and } U \leq \eta \\
    1, & \frac{dP_{i,j}}{dQ_{i,j}} <  \frac{w_{i,j}^-}{w_{i,j}^+}\gamma \\
    1, & \frac{dP_{i,j}}{dQ_{i,j}} =  \frac{w_{i,j}^-}{w_{i,j}^+}\gamma \text{ and } U > \eta
    \end{cases}
\end{equation}
where $\gamma$ and $\eta\in[0,1]$ are chosen so that $\epsilon^+ = \epsilon$.
\end{theorem}
\begin{proof}
The proof is given in Section Appendix \ref{proof::optimal}
\end{proof}
\begin{figure}
    \centering
    \includegraphics[scale=.35]{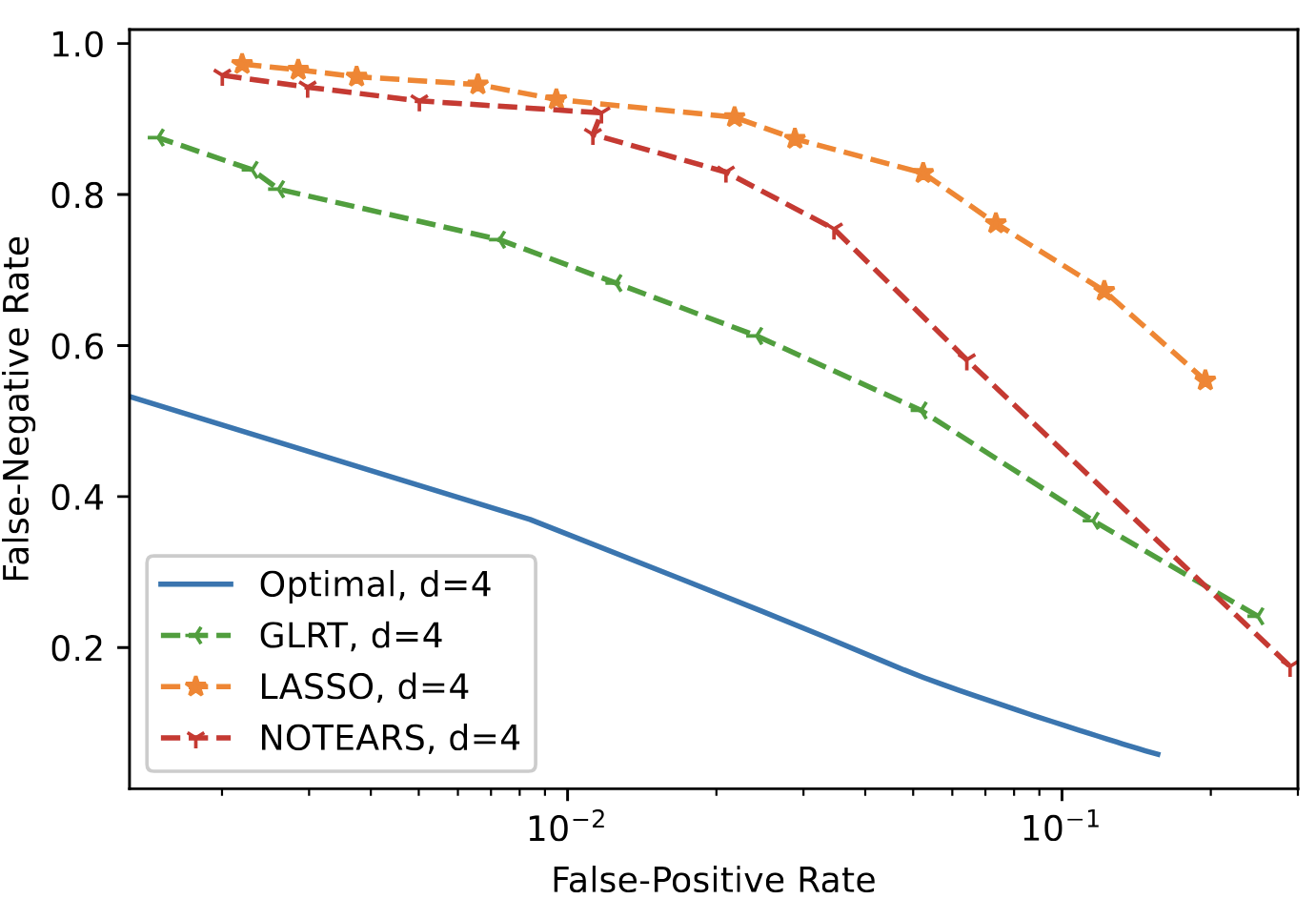}
    \caption{{Performance comparison showing the large gap between the optimal detector and several state-of-the-art methods on $n=10$ observations with $\sigma^2 = 1$.}}
    \label{fig:optimal}
\end{figure}
Similar to the optimal detector in Neyman-Pearson hypothesis testing, the random variable $U$ controls our randomization to achieve the false negative rate exactly 
{\footnote{If the priors have different support for each hypothesis or the observations are discrete valued, randomization may be necessary per classical Neyman-Pearson hypothesis testing, \emph{e.g.}
\cite{poor2013introduction,van2004detection}.}.} The proof of Theorem \ref{optimalrule} is similar to that of the classical Neyman-Pearson lemma (Proposition II.D.1 in \cite{poor2013introduction}) and is given in Appendix \ref{proof::optimal}. {In Figure \ref{fig:optimal}, we compare the performance of the optimal detector with various state-of-the-art methods on $n=10$ observations with $\sigma^2 =1$. A description of each method is provide in Section \ref{Sec::prioralgorithms}. For the graph prior $\pi_{\boldsymbol{A}}$, let $\mathcal{D}$ denote the set of all matrices with entries equal to either 0 or 1 corresponding to a directed acyclic graph with $d$ vertices. For Figure \ref{fig:optimal}, we uniformly, at random, select a matrix from $\mathcal{D}$. The uniform graph prior and optimal detector can be computed for small graph sizes, but become expensive to compute for larger graphs. Thus, in the sequel, Section \ref{Sec::prioralgorithms},  we provide an alternative algorithm with finite sample performance guarantees whose complexity has a more favorable scaling with graph size.  In addition,  a lower complexity graph prior as an approximation to the uniform prior is also considered.
 Perhaps surprisingly, Figure \ref{fig:optimal} shows that the optimal detector strongly outperforms other methods, especially in the low false-positive regime. This gap suggests the need to develop improved finite-sample methods for low false-positive tolerance levels.
\par An interesting observation is that the threshold $\gamma$ and the randomization $\eta$ \textit{do not depend on the edge being considered}. The effect of the edge in question is captured by the edge weights in the factor $\frac{w_{i,j}^-}{w_{i,j}^+}$ (which is completely determined by the prior $\pi_{\boldsymbol{A}}$). This feature simplifies design and analysis.
 For example, the next proposition follows relatively easily due to the form of $\hat{\boldsymbol{\chi}}^*$.
\begin{proposition}\label{upperbound}
Assume $\boldsymbol{A}\sim \pi_{\boldsymbol{A}}$. For the detector $\hat{\boldsymbol{\chi}}^*$ given in Proposition \ref{optimalrule}, we have, for any $\lambda\in[0,1]$
\begin{equation}
    \epsilon^- \leq \sum_{i,j} (w^-_{i,j})^{1-\lambda} (w^+_{i,j})^{\lambda} \frac{1}{\gamma^{\lambda}} e^{-(1-\lambda)D_{\lambda}(P_{i,j}||Q_{i,j})},
\end{equation}
where $D_{\lambda}(P_{i,j}||Q_{i,j})$ is the  the R\'enyi divergence of order $\lambda$ between the series of distributions $P_{i,j}$ and $Q_{i,j}$ which is defined in Appendix \ref{AppenDefinitions}. 
\end{proposition}
\begin{proof}
    The proof is given in Appendix \ref{ProofofUpper}
\end{proof}
The next theorem is a converse bound for any detector $\hat{\boldsymbol{\chi}}$.
\begin{proposition}\label{converse}
    Assume $\boldsymbol{A}\sim \pi_{\boldsymbol{A}}$. For any detector $\hat{\boldsymbol{\chi}}$ that satisfies $\epsilon^+ \leq \epsilon$, we have that for any $\lambda\in[0,1]$.
    \begin{equation}
        \begin{aligned}\label{conversebound}
        \epsilon^- &\geq \frac{1}{2} \sum_{i,j}w^-_{i,j}\exp\big\{-(1-\lambda)D_{\lambda}(P_{i,j}||Q_{i,j}) 
        \\
        -&\lambda D'_{\lambda}(P_{i,j}||Q_{i,j})-\lambda \sqrt{2D_{\lambda}''(P_{i,j}||Q_{i,j})}\big\} 
        \\ 
        & \qquad -\epsilon \max_{i,j}\bigg\{ \frac{w^-_{i,j}}{w^+_{i,j}}\exp\big\{-D'_{\lambda}(P_{i,j}||Q_{i,j}) 
        \\
        & \qquad +(1-2\lambda) \sqrt{2D_{\lambda}''(P_{i,j}||Q_{i,j})}\big\}\bigg\},
        \end{aligned}
    \end{equation}
    where $D_{\lambda}'(P_{i,j}||Q_{i,j})$ and $D_{\lambda}''(P_{i,j}||Q_{i,j})$ are defined in Appendix \ref{AppenDefinitions}.
\end{proposition}
\begin{proof}
    The proof is given in Appendix \ref{ProofofConverse}
\end{proof}
Propositions \ref{upperbound} and \ref{converse} show that the  R\'enyi divergence controls the false negative rate for a given tolerance $\epsilon$. There are some important notes about Theorem \ref{optimalrule} and Propositions \ref{upperbound} and \ref{converse}.
\begin{enumerate}
    \item The detector specified in Theorem 1 is optimal, but is prior, $\pi_{\boldsymbol{A}}$, dependent as the computations of  
    $P_{i,j}$ and $Q_{i,j}$ do depend on $\pi_{\boldsymbol{A}}$.
    \item Theorem 1 is general. The proof of Theorem 1 makes no assumptions on the specific model class, and hence holds for any general class of models, such as differentiable systems with additive noise, linear systems with non-Gaussian noise, \emph{etc.} This is \textit{not} to say that the detector given by Theorem 1 is necessarily \textit{easy to compute} for any system model, as some models are more computationally tractable (such as the linear Gaussian case) than others. The theorem does provide insight into how to control error rates for the general causal discovery problem (and hence any general class of model).
    \item It is shown in \cite{ISITPaper} that the proposed bounds are tighter for small $n$ and small false positive tolerance than previously considered bounds for the same metrics, such as those presented in \cite{HayekLower}.
    \item  As expected, the R\'enyi divergence provides tighter bounds in the finite sample regime than classical Chernoff bounds 
    and further provides an explicit dependence on $\epsilon$. As such, for appropriately selected $\lambda$, the bound given in \eqref{conversebound} reduces to an expression akin to classical bounds given in \cite{SGB}. 
\end{enumerate}

Unfortunately, although the form of the optimal detector $\hat{\boldsymbol{\chi}}^*$ is simple to express, in practice, it is often computationally intractable, even for small graphs and relatively simple priors.
\par
To see this, consider a system with two vertices, see Figure~\ref{fig:graphs}. We receive observation $\{X_k\}_1^n$ with $X_k = [X_{k}(1), X_{k}(2)]^\top$, $k=1,2,...,n$. As stated before, we assume an LSEM,
\begin{equation}\label{TwoNodeSystem}
    \begin{bmatrix}
    X_{k}(1) \\
    X_{k}(2)
    \end{bmatrix}
    =
    \boldsymbol{A}\begin{bmatrix}
    X_{k}(1) \\
    X_{k}(2)
    \end{bmatrix}
    +
    \begin{bmatrix}
    W_{k}(1) \\
    W_{k}(2)
    \end{bmatrix},
\end{equation}
where the $[W_{k}(1), W_{k}(2)]^\top$ vectors are \textit{i.i.d.} Gaussian vectors with zero mean and covariance matrix $\sigma^2 I$. Since we restrict ourselves to directed acyclic graphs, the adjacency matrix $\boldsymbol{A}$ can only take one of three possible forms, which we denote as follows,
\begin{equation}
    \boldsymbol{A}_0 = 
    \begin{bmatrix}
        0 & 0 \\
        0 & 0
    \end{bmatrix},
    \qquad 
    \boldsymbol{A}_1(a) = 
    \begin{bmatrix}
        0 & 0 \\
        a & 0
    \end{bmatrix},
    \qquad
    \boldsymbol{A}_2(a) = 
    \begin{bmatrix}
        0 & a \\
        0 & 0
    \end{bmatrix},
\end{equation}
where $a\in\mathbb{R}\setminus\{0\}$. Assume that the prior $\pi_{\boldsymbol{A}}$ selects from the structures of $\boldsymbol{A}_0$, $\boldsymbol{A}_1(a)$, and $\boldsymbol{A}_2(a)$ uniformly at random, and that $a=1+\zeta$ where $\zeta$ is a standard exponential random variable. In order to implement $\hat{\boldsymbol{\chi}}^*$ we must first compute the conditional distributions $P_{1,2}$ and $Q_{1,2}$. To compute $P_{1,2}$, observe that if $\boldsymbol{\chi}_{1,2} = 0$, then $\boldsymbol{A}_0$ must be the true graph structure, and so $X_{k}(1)$ and $X_{k}(2)$ are simply \textit{i.i.d} Gaussian random variables with variance $\sigma^2$. To compute $Q_{1,2}$, notice that if  $\boldsymbol{\chi}_{1,2} = 1$, the true graph may correspond to either $\boldsymbol{A}_1$ or $\boldsymbol{A}_2$. In either case, $X_{1,k}$ and $X_{2,k}$ are zero-mean jointly Gaussian random variables with covariance matrix $\sigma^2 (I-\boldsymbol{A}_1(a))^{-1}(I-\boldsymbol{A}_1(a))^{-\top}$ under $\boldsymbol{A}_1(a)$, and $\sigma^2(I-\boldsymbol{A}_2(a))^{-1}(I-\boldsymbol{A}_2(a))^{-\top}$ under $\boldsymbol{A}_2(a)$. Then, we have that
\begin{equation}
   \begin{aligned}\label{eq:complicated_marginal}
        &Q_{1,2}(X_k) =
        \\
        &\frac{1}{2}\int_1^\infty \sum_{i=1}^2\frac{e^{-\frac{1}{2\sigma^2}X_k^\top (I-\boldsymbol{A}_i(a))^{\top}(I-\boldsymbol{A}_i(a))X_k}}{2\pi |\sigma^2 (I-\boldsymbol{A}_i(a))^{-1}(I-\boldsymbol{A}_i(a))^{-\top}|} e^{-(a-1)}da.
   \end{aligned}
\end{equation}
Moreover, observe that this is only for one observation, and that the joint distribution becomes more complex since the observations are \textit{not independent a priori} and are only independent conditioned on $\boldsymbol{A}$. That is, for the joint distribution, the terms in the sum of \eqref{eq:complicated_marginal} become products of Gaussian distributions. From this example, it is not difficult to see how computing the distributions $P_{i,j}$ and $Q_{i,j}$ can become computationally infeasible as the number of vertices increases, or as the graph prior $\pi_{\boldsymbol{A}}$ becomes more complex.

\section{$\epsilon$-CUT}\label{Sec::CUT}
\begin{algorithm*}[t]
\caption{$\epsilon$-CUT}\label{algo}
Specify user false alarm tolerance $0<\epsilon\leq 1$ and the variance of the noise terms $\sigma^2$. Let $2^{\mathcal{V}}$ denote the power set of $\mathcal{V}$ (the vertex set). Let $\mathcal{B}$ denote the set of all unique pairs of edges, i.e., the set of all pairs $(i,j)$, such that $i\neq j$, $i,j\in\mathcal{V}$. The cumulative distribution function (cdf) of a chi-squared distribution with $l$ degrees of freedom is denoted by $F_l(x)$.

For each $(i,j)\in\mathcal{B}$:
    \begin{enumerate}
        \item For each $\hat{\mathcal{Z}}(i)\in2^{\mathcal{V}}$ and $\hat{\mathcal{Z}}(j)\in2^{\mathcal{V}}$ with $j\notin \hat{\mathcal{Z}}(i)$ and $i\notin \hat{\mathcal{Z}}(j)$ (i.e., the potential parent sets of nodes $i$ and $j$):
        \begin{enumerate}
            \item Perform linear least squares regression for nodes $i$ and $j$ on $\hat{\mathcal{Z}}(i)$ and $\hat{\mathcal{Z}}(j)$, respectively, obtaining the vectors of coefficients $\hat{\alpha}$ and $\hat{\beta}$, i.e.,
            \begin{equation*}
                \hat{\alpha} = (\hat{\boldsymbol{Z}}(i)\hat{\boldsymbol{Z}}(i)^\top)^{-1}\hat{\boldsymbol{Z}}(i)^\top \boldsymbol{X}(i)
            \end{equation*}
            where $\hat{\boldsymbol{Z}}(i) = [\hat{Z}_1(i),\hat{Z}_2(i),...,\hat{Z}_n(i)]^\top$, $\hat{Z}_k(i) = [X_k(i_1), X_k(i_2),...,X_k(i_{|\hat{\mathcal{Z}}(i)|})]^\top$ for $i_l \in \hat{\mathcal{Z}}(i)$, $l=1,2,...,|\hat{\mathcal{Z}}(i)|$, and $\boldsymbol{X}(i) = [X_1(i),X_2(i),...,X_n(i)]^\top$. Similarly for $\hat{\beta}$.
            \item Compute
            \begin{align*}
                \hat{\sigma}^2_i = \sum_{k=1}^n \big(X_k(i) - \hat{\alpha}^\top \hat{Z}_k(i)\big)^2, \quad
                \hat{\sigma}^2_j = \sum_{k=1}^n \big(X_k(j) - \hat{\beta}^\top \hat{Z}_k(j)\big)^2,
            \end{align*}
            i.e., the residual sum of squares for vertices $i$ and $j$.
            \item Define $ \tau'_{\hat{\mathcal{Z}}(i),\hat{\mathcal{Z}}(j)} =  (\tau_{\hat{\mathcal{Z}}(i),\hat{\mathcal{Z}}(j)}- |q-p|\sigma^2)/(2\sigma^2)$ and compute $\tau_{\hat{\mathcal{Z}}(i),\hat{\mathcal{Z}}(j)}$ such that
            \begin{align*}
                &2-F_{n-p}(\tau'_{\hat{\mathcal{Z}}(i),\hat{\mathcal{Z}}(j)} + n- p) + F_{n-p}(-\tau'_{\hat{\mathcal{Z}}(i),\hat{\mathcal{Z}}(j)} + n- p)
                \\
                &\qquad \qquad \qquad -F_{n-q}(\tau'_{\hat{\mathcal{Z}}(i),\hat{\mathcal{Z}}(j)} + n- q) + F_{n-q}(-\tau'_{\hat{\mathcal{Z}}(i),\hat{\mathcal{Z}}(j)} + n- q)=\epsilon,\label{Fbound}
            \end{align*}
            where $p = |\hat{\mathcal{Z}}(i)|$ and $q=|\hat{\mathcal{Z}}(j)|$.
            \item If $|\hat{\sigma}^2_i - \hat{\sigma}^2_j| \leq  \tau_{\hat{\mathcal{Z}}(i),\hat{\mathcal{Z}}(j)}$, declare $\hat{\chi}_{i,j}=0$.
        \end{enumerate}
        \item If all potential parent sets $\hat{\mathcal{Z}}(i)\in2^{\mathcal{V}}$ and $\hat{\mathcal{Z}}(j)\in2^{\mathcal{V}}$ have been tested without declaring $\hat{\chi}_{i,j}=0$, declare $\hat{\chi}_{i,j}=1$. 
    \end{enumerate}
\end{algorithm*}
Because of the shortcomings mentioned above, we derive an alternative algorithm, the  \textit{ $\epsilon$-rate \underline{C}ausal discovery with \underline{U}nequal edge error \underline{T}olerance} ($\epsilon-CUT$) algorithm.Recall that detecting the absence or presence of an edge is equivalent to a binary hypothesis test (see \eqref{HypoTest0} and \eqref{HypoTest1}) for which we provide an alternative algorithm. This algorithm relies on the fact that children have higher variances than their parents, which we highlight through an illustrative example below. In addition, we provide time complexity analysis for $\epsilon-CUT$. 
\par
Via an example, we provide the intuition for our new detector.
Assume we have the same two-vertex given by equation \eqref{TwoNodeSystem}, with the understanding that $a$ may be any non-zero real number (which may be either random or deterministic). If $\boldsymbol{A}=\boldsymbol{A}_0$, we have that for all $k$,
\begin{align}
    \mathbb{E}[X_k(1)^2] &= \mathbb{E}[W_k(1)^2] = \sigma^2, \\
    \mathbb{E}[X_k(2)^2] &= \mathbb{E}[W_k(2)^2] = \sigma^2.
\end{align}
Alternatively, if $\boldsymbol{A}=\boldsymbol{A}_1$, we still have $\mathbb{E}[X_k(1)^2] = \sigma^2$, but
\begin{equation}
    \mathbb{E}[X_k(2)^2] = \mathbb{E}[\big(a X_k(1) + W_k(2)\big)^2] = (a^2 + 1)\sigma^2.
\end{equation}
Similarly, if $\boldsymbol{A}=\boldsymbol{A}_2$, then we have
\begin{align}
    \mathbb{E}[X_k(1)^2] = (a^2 + 1)\sigma^2, \qquad 
    \mathbb{E}[X_k(2)^2] =  \sigma^2.
\end{align}
Hence, we can study the following equivalent hypothesis testing problem,
\begin{align}
    H_0: \mathbb{E}[X_k(1)^2] = \mathbb{E}[X_k(2)^2], \\
    H_1: \mathbb{E}[X_k(1)^2] \neq \mathbb{E}[X_k(2)^2],
\end{align}
since, in the linear Gaussian setting, any solution to the above problem is also a solution to \eqref{HypoTest0}-\eqref{HypoTest1}.
\par
Our hypothesis test is given as follows. If $H_0$ is true, then the empirical estimates of the variances should be roughly equal with high probability. That is, for a properly specified threshold $\tau$, the event $|\hat{\sigma}^2_1 - \hat{\sigma}^2_2| \leq  \tau$,
where
\begin{equation}
    \hat{\sigma}^2_1 = \sum_{k=1}^n X_k(1)^2, \quad \hat{\sigma}^2_2 = \sum_{k=1}^n X_k(2)^2,
\end{equation}
should occur with high probability. To extend the intuition to the case of more than two vertices, notice that if we condition on the parents of vertex $i$, $\hat{\sigma}^2_i$ is now given by
\begin{align}
    \hat{\sigma}^2_i &= \sum_{k=1}^n \big(X_k(i) - \alpha_i^\top Z_k(i)\big)^2, \label{estimate}
\end{align}
where $\alpha_i$ is the vector of non-zero edge weights in the $i$th row of $\boldsymbol{A}$.
{Recall that tests based on the residual sum of squares has been previously considered\cite{residuals_latent,Zhang_Zhou_Guan_residuals}. However, these works test {\em independence} between residuals and only provide asymptotic guarantees. In contrast, our algorithm compares the difference between residuals to detect energy differences. Moreover, we can select the threshold $\tau$ to satisfy the false positive constraint in a \textit{finite-sample setting}}. The following lemma, whose proof is in Appendix \ref{ProofChiSquared}, enables the setting of the needed threshold to achieve our finite-sample constraints.
\begin{lemma}\label{chisquaredlemma}
For any $n$ and $i\in\mathcal{V}$ let $\mathcal{Z}(i)$ and $Z_k(i)$ be defined as in the top of Section \ref{Sec::problemformulation}. Define
\begin{equation}
    \hat{\sigma}^{*2}_i = \sum_{k=1}^n \big(X_k(i) - \hat{\alpha}_i^\top Z_k(i)\big)^2,
\end{equation}
where $\hat{\alpha}_i$ is the vector of coefficients resulting from performing least squares for vertex $i$ on $\{Z_k(i)\}_{k=1}^n$. Then, conditioned on $\{Z_k(i)\}_{k=1}^n$, $\hat{\sigma}^{*2}_i/\sigma^2$ follows a chi-squared distribution with $n-p$ degrees of freedom, where $p = |\mathcal{Z}(i)|$.
\end{lemma}
The complete algorithm for $\epsilon-CUT$ is given in Algorithm \ref{algo}. With Lemma \ref{chisquaredlemma}, we can prove the following result.
\begin{theorem}\label{satisfyconstraint}
Assume we have a LSEM, and that $\boldsymbol{A}\sim \pi_{\boldsymbol{A}}$. Then, for any number of samples $n$ and any given false positive tolerance $\epsilon$, the false positive rate of $\epsilon$-CUT, denoted by $\epsilon^+_1$, satisfies $\epsilon^+_1 \leq \epsilon$.
\end{theorem}
Some important notes regarding $\epsilon$-CUT and Theorem \ref{satisfyconstraint}.
\begin{enumerate}

    \item Theorem \ref{satisfyconstraint} is a finite-sample result. This differs from much of the current literature which focuses on asymptotic recoverability guarantees of causal discovery \cite{chickering2002learning,equivariance,StriederDrtonconfidence,causalmatrixcompletion,pearl2009causality}. Finite-sample results have appeared in the literature, with a notable result appearing in \cite{lasso} (Theorem 3). However, the result in \cite{lasso} deals with connected components instead of individual edges, which is our main objective.
    \item The exact algorithm for $\epsilon-CUT$, outlined in Algorithm \ref{algo}, does not depend on the choice of prior $\pi_{\boldsymbol{A}}$. Hence, $\epsilon-CUT$ may be used in either a Bayesian or non-Bayesian setting.
    \item The procedure outlined in Algorithm \ref{algo} requires finding a different threshold for each value of $p$ and $q$ given in Algorithm \ref{algo} c). While this is not computationally difficult it is possible to use only a single threshold in Algorithm \ref{algo} while maintaining the guarantee that the false positive rate is less than $\epsilon$. To see this, observe that for two chi-squared random variables with $p$ and $q$ degrees of freedom, respectively, where $p\leq q$ we have that for any $x\in\mathbb{R}$, $F_p(x) \geq F_q(x)$, where $F_p$ is the cdf of a chi-squared random variable with $p$ degrees of freedom. This tells us that the term in Algorithm \ref{algo} c) is upper-bounded by
    \begin{equation}
        \begin{aligned}
2\big[1-F_n (\tau'+n) + F_{n-d+2}(-\tau' +n -d + 2)\big].
        \end{aligned}
    \end{equation}
    where $\tau' = (\tau - (d-2)\sigma^2)/(2\sigma^2)$. Hence, using only a single threshold while maintaining the false positive constraint is possible.
    \item As mentioned before, different solutions to the hypothesis-testing problem between $H_0$ and $H_1$, lead to different causal discovery algorithms. Another way to circumvent the computational challenges associated with the optimal detector $\hat{\boldsymbol{\chi}}^*$ is to use a generalized likelihood ratio test (GLRT). This is done in \cite{StriederDrtonconfidence} to obtain confidence intervals on the causal effect between two vertices, rather than directly declaring the presence or absence on an edge. Unfortunately, the main results in \cite{StriederDrtonconfidence} are asymptotic. Hence, $\epsilon$-CUT circumvents the computational issues associated with computing $\hat{\boldsymbol{\chi}}^*$ while providing finite-sample results.
    \item Unlike several popular algorithms such as BIC \cite{equivariance}, LASSO \cite{lasso}, and NOTEARS \cite{NOTEARS}, $\epsilon$-CUT requires no hyper-parameter tuning. That is, in the algorithms mentioned, a sparsity constraint is added through a regularization term ($\ell_0$ in \cite{equivariance} and  $\ell_1$ in \cite{lasso,NOTEARS}). The constant $\lambda$ controls the sparsity of the resulting graph. Unfortunately, there is no current way to determine \textit{a priori} how the constant $\lambda$ affects the error rates, and if these rates satisfy the constraint in \eqref{NPCI}. Hence, one needs to experiment with different regularization constants. In contrast $\epsilon$-CUT, once $\epsilon$ is given, everything in $\epsilon$-CUT is completely specified.
    \end{enumerate}
    \vspace*{-0.1in}
\begin{proof}[Proof Sketch of Theorem \ref{satisfyconstraint}]
Since $\epsilon^+ = \sum_{i,j} w^+_{i,j}P_{i,j}^+$, it suffices to show that for all $i,j$, $P^+_{i,j} \leq \epsilon$. Note that
\vspace*{-0.1in}
\begin{equation}
    \begin{aligned}
P^+_{i,j} &= \int_{\boldsymbol{A}}\int_{\boldsymbol{X}^{\setminus i,j}}\mathbb{P}_{\boldsymbol{A}}(\hat{\boldsymbol{\chi}}_{i,j}=1|\boldsymbol{X}^{\setminus i,j}) \mathbb{P}_{\boldsymbol{A}}(\boldsymbol{X}^{\setminus i,j}) 
\\
&\qquad \qquad \qquad \qquad \frac{\mathbbm{1}\{\boldsymbol{A}_{i,j}=0 \cap \boldsymbol{A}_{j,i}=0\}\pi_{\boldsymbol{A}}}{\mathbb{P}(\boldsymbol{\chi}_{i,j}=0)},
    \end{aligned}
\end{equation}
where $\boldsymbol{X}^{\setminus i,j}$ denotes the set of all measurements except those from the $i$th and $j$th vertex. That is, $\boldsymbol{X}^{\setminus i,j} = \{X^{\setminus i, j}_k\}_{k=1}^n$ where $X^{\setminus i, j}_k = [X_k(1),...,X_k(i-1),X_k(i+1),...,X_k(j-1), X_k(j+1),...,X_k(d)]^\top$.
Then, it suffices to show that for any $\boldsymbol{A}$ and $\boldsymbol{X}^{\setminus i,j}$ we have $\mathbb{P}_{\boldsymbol{A}}(\hat{\boldsymbol{\chi}}_{i,j}=1|\boldsymbol{X}^{\setminus i,j}) \leq \epsilon$.
This is done by first noticing that for $\hat{\sigma}^2_i$ and $\hat{\sigma}^2_j$ as defined in Algorithm \ref{algo},
\begin{align}
    &\mathbb{P}_{\boldsymbol{A}}(\hat{\boldsymbol{\chi}}_{i,j}=1|\boldsymbol{X}^{\setminus i,j}) 
\end{align}
\begin{align}
    &= \mathbb{P}_{\boldsymbol{A}}\bigg(\bigcap_{\hat{\mathcal{Z}}(i),\hat{\mathcal{Z}}(j)}|\hat{\sigma}^2_i - \hat{\sigma}^2_j| >  \tau_{\hat{\mathcal{Z}}(i),\hat{\mathcal{Z}}(j)}|\boldsymbol{X}^{\setminus i,j}\bigg)\\
    &\overset{(a)}{\leq} \mathbb{P}_{\boldsymbol{A}}\big(|\hat{\sigma}^{*2}_i - \hat{\sigma}^{*2}_j| >  \tau_{\mathcal{Z}(i),\mathcal{Z}(j)}|\boldsymbol{X}^{\setminus i,j}\big)\label{goalprob},
\end{align}

where $(a)$ holds since the event $|\hat{\sigma}^2_i - \hat{\sigma}^2_j| >  \tau_{\hat{\mathcal{Z}}(i),\hat{\mathcal{Z}}(j)}$ must hold for \textit{all potential parent sets} $\hat{\mathcal{Z}}(i),\hat{\mathcal{Z}}(j)$, including the true parent sets $\mathcal{Z}(i),\mathcal{Z}(j)$, and so 
\begin{equation}
    \begin{aligned}
        \bigg\{\bigcap_{\hat{\mathcal{Z}}(i),\hat{\mathcal{Z}}(j)}|\hat{\sigma}^2_i - \hat{\sigma}^2_j| >  \tau_{\hat{\mathcal{Z}}(i),\hat{\mathcal{Z}}(j)}\bigg\} 
        \\
        \subset \{|\hat{\sigma}^{*2}_i - \hat{\sigma}^{*2}_j| >  \tau_{\mathcal{Z}(i),\mathcal{Z}(j)}\}.
    \end{aligned}
\end{equation}
Using Lemma \ref{chisquaredlemma} together with a series of inequalities and algebraic manipulations, \eqref{goalprob} is upper bounded by the expression given in Algorithm \ref{algo} c), which completes the proof. The full proof is given in Appendix \ref{proof::eCUT}. 
\end{proof}
Notice that in the proof of Theorem \ref{satisfyconstraint} it is shown that for any $\boldsymbol{A}$ (with both $\boldsymbol{A}_{i,j}=0 $ and $\boldsymbol{A}_{j,i}=0$ and $\boldsymbol{X}^{\setminus i,j}$ we have $\mathbb{P}_{\boldsymbol{A}}(\hat{\boldsymbol{\chi}}_{i,j}=1|\boldsymbol{X}^{\setminus i,j}) \leq \epsilon$. This then implies that
\begin{align}
    \mathbb{P}_{\boldsymbol{A}}(\hat{\boldsymbol{\chi}}_{i,j}=1) &=\int_{\boldsymbol{X}^{\setminus i,j}}\mathbb{P}_{\boldsymbol{A}}(\hat{\boldsymbol{\chi}}_{i,j}=1|\boldsymbol{X}^{\setminus i,j})\mathbb{P}_{\boldsymbol{A}}(\boldsymbol{X}^{\setminus i,j}) \\
    &\leq \epsilon \int_{\boldsymbol{X}^{\setminus i,j}}\mathbb{P}_{\boldsymbol{A}}(\boldsymbol{X}^{\setminus i,j}) = \epsilon.
\end{align}
Hence, if the true underlying graph does not have an edge between $i$ and $j$, the probability of a false positive does not exceed $\epsilon$. Observe that this is a non-Bayesian result, and so $\epsilon-CUT$ may be used in a Bayesian or non-Bayesian setting while preserving finite-sample results.

    \subsection{Complexity Analysis}
For each pair of vertices $i$ and $j$, we must conduct a threshold test between the residual sum of squares for potentially every pair of vertex subsets that do not contain $i$ or $j$. Since there are $4^{d-2}$ possible pairs of these subsets for each pair of vertices, we must compute at most $\frac{d(d-1)}{2}4^{d-2}$ linear regressions. The computational complexity for linear regression is $\mathcal{O}(d^2(n+d))$ (the details are in the supplemental file located at https://github.com/shaskajo/epsilonCUT). For each pair of subsets, we must also perform a threshold test on the difference between the residuals, which has complexity $\mathcal{O}(1)$. Hence, the overall complexity of $\epsilon$-CUT is $\mathcal{O}(d^3(d-1)4^{d-2}(n+d))$.

\section{Comparison Algorithms}\label{Sec::prioralgorithms}
In Section \ref{Sec::numerical} we consider numerical examples to compare $\epsilon$-CUT to other popular algorithms such as NOTEARS \cite{NOTEARS}, DAGMA \cite{DAGMA}, LASSO neighborhood selection \cite{lasso}, and the generalized likelihood ratio test \cite{StriederDrtonconfidence}. Hence, we briefly summarize all prior algorithms used in the sequel. We also briefly discuss the computational complexity of LASSO and the generalized likelihood ratio test. We focus on these algorithms since they possess some form of finite-sample guarantee of recovering the underlying graph. Moreover, continuous optimization formulations, such as NOTEARS and DAGMA, suffer from a highly non-convex optimization space, and hence cannot guarantee convergence to global optima for finite samples, as well as making computational complexity analysis rather difficult and beyond the scope of this paper. However, we provide runtime results in Section \ref{Sec::numerical} for all algorithms used. Because of this, we also briefly mentioned how each method is implemented.
\subsection{LASSO}
For a given observed node $X_{k}(i)$, let $Y^i_k=[X_{k}(1), ...,X_{k}(i-1), X_{k}(i+1), ... , X_{k}(d)]^\top$ for $k=1,2,...,n$. Then, assuming a linear model for the observations, $\hat{X}_{k}(i) = A_i^\top Y_{i,k}$, where $A_i$ is a vector of coefficients, LASSO \cite{lasso} minimizes the following  sparsity penalized loss over the coefficients of $A_i$,
    \begin{equation}
        \frac{1}{2n} \sum_{k=1}^n \|X_{k}(i) - A_i^\top Y^i_k \|_2^2 + \lambda \|A_{i} \|_1,
    \end{equation}
    for each node $i$, where $\lambda$ is a regularizer term.
    \par 
    The complexity of LASSO neighborhood selection for a single vertex is $\mathcal{O}(nd\min\{n,d\})$ \cite{lasso}. Hence, the total complexity for LASSO is $\mathcal{O}(nd^2\min\{n,d\})$. We implement LASSO using the Python implementation provided in the sklearn library.
\subsection{Generalized Likelihood Ratio Test (GLRT)}
We consider a generalized likelihood ratio test (GLRT) for our hypothesis testing problem. \cite{StriederDrtonconfidence}. In \cite{StriederDrtonconfidence}, the GLRT is used to construct confidence intervals on the causal effect rather than detect the presence or absence of an edge. We modify the algorithm  in \cite{StriederDrtonconfidence} to address our problem. If we let $\mathcal{M}$ denote the set of all DAGS and $\mathcal{M}_{i,j}^0$ the set of all DAGS with no edge between vertices $i$ and $j$. Then, for given DAG $\mathcal{G}$ with weighted adjacency matrix $\boldsymbol{A}$, the likelihood function is given as
    \begin{equation}\label{likelihood}
        \ell (\mathcal{G}) = -\frac{np}{2}\log(2\pi\sigma^2) - \frac{n}{2\sigma^2}\text{Trace}\big((I-\boldsymbol{A})^\top(I-\boldsymbol{A})\hat{\Sigma} \big),
    \end{equation}
    where $\hat{\Sigma}$ is the empirical covariance matrix. Then, define the test statistic
    \begin{equation}
        \Lambda_{i,j} = 2(\sup_{\mathcal{G}\in\mathcal{M}} \ell (\mathcal{G}) - \sup_{\mathcal{G}\in\mathcal{M}_{i,j}^0} \ell (\mathcal{G})),
    \end{equation}
    which is compared against a threshold $\tau$. If $\Lambda_{i,j} > \tau$ then an edge is declared, otherwise declare $\hat{\boldsymbol{\chi}}_{i,j}=0$. The threshold $\tau$ is chosen so that $\Lambda_{i,j} > \tau$ occurs with probability $\epsilon$ when $\boldsymbol{\chi}_{i,j}=0$ \textit{asymptotically}. We note that reducing the computational complexity of the inherent exhaustive DAG search for ML/GLRT schemes remains an open problem.
    \par
     To obtain the computational complexity of the GLRT, we use the procedure outlined in Section 5 of \cite{StriederDrtonconfidence}, with slight modifications to fit our problem (since our algorithm outputs a potential edge rather than a confidence interval), and can be found in the following GitHub repository: https://github.com/shaskajo/epsilonCUT/tree/main. The details of our implementation along with the full analysis of the computational complexity are given in the supplemental material: https://github.com/shaskajo/epsilonCUT. The computational complexity of the GLRT is $\mathcal{O}((1+d!)d^3(d-1)^2(n+d))$.
\subsection{Continuous Optimization Formulations}
A popular framework used in causal discovery is mapping the combinatorial optimization problem into a continuous one. Typically, a loss function is minimized subject to a constraint that ensures the outputted matrix corresponds to a DAG. Throughout this paper, we assume the loss function is the penalized squared loss,
 \begin{equation}
        \frac{1}{2n}\sum_{k=1}^n \|X_k - \boldsymbol{A}X_k \|_2^2 + \lambda \|\boldsymbol{A} \|_1,
    \end{equation}
    where $\lambda$ is a regularizer term, $\|X_k - \boldsymbol{A}X_k \|_2$ is the $l_2$ norm of the vector $X_k - \boldsymbol{A}X_k$ and $\|\boldsymbol{A} \|_1$ is the $l_1$ norm of the matrix $\boldsymbol{A}$. Then, the following algorithms differ only in what algebraic constraint is used to ensure $\boldsymbol{A}$ corresponds to a DAG.
\begin{enumerate}
    \item \textit{NOTEARS}: It is shown in \cite{NOTEARS} that a weighted adjacency matrix $\boldsymbol{A}$ represents a directed acyclic graph if and only if
    \begin{equation}
        \text{Trace}(e^{\boldsymbol{A}\circ \boldsymbol{A}}) = d,
    \end{equation}
    where $\circ$ denotes the Hadamard product (element-wise multiplication).  The elements of the continuous-valued estimate $\hat{\boldsymbol{A}}$ are thresholded to determine the inactive edges, note that this threshold is another hyperparamter to be tuned.
     We use the implementation publicly provided in \cite{NOTEARS}.
    \item DAGMA \cite{DAGMA} solves the continuous optimization problem subject to the constraint
    \begin{equation}
        -\log |s\boldsymbol{I}-\boldsymbol{A}\circ \boldsymbol{A}| + d\log s = 0, \label{DAGMAconstraint}
    \end{equation}
    where $s>0$ is a hyperparameter. Similarly to NOTEARS, \cite{DAGMA} shows that $\boldsymbol{A}$ is a DAG if and only if \eqref{DAGMAconstraint} holds for appropriately selected $s$. We use the implementation publicly provided in \cite{DAGMA}.
\end{enumerate}

\begin{table}[t]
    \centering
    \begin{tabular}{ |p{2cm}||p{5cm}|  }
 \hline
 \multicolumn{2}{|c|}{Computational Complexities} \\
 \hline
 Algorithm& Complexity\\
 \hline
 LASSO   & $\mathcal{O}(d^2n\min\{n,d\})$  \\
$\epsilon$-CUT&   $\mathcal{O}(d^3(d-1)4^{d-2}(n+d))$\\
 GLRT& $\mathcal{O}((d!+1)d^3(d-1)^2(n+d))$ \\
 \hline
\end{tabular}
    \caption{Computational complexities of LASSO, $\epsilon$-CUT, and the GLRT.}
    \label{tab:complexities}
\end{table}

\section{Numerical Results}\label{Sec::numerical}
\begin{figure*}
    \centering
    \begin{subfigure}{0.4\textwidth}
    \centering
    \includegraphics[scale=.35]{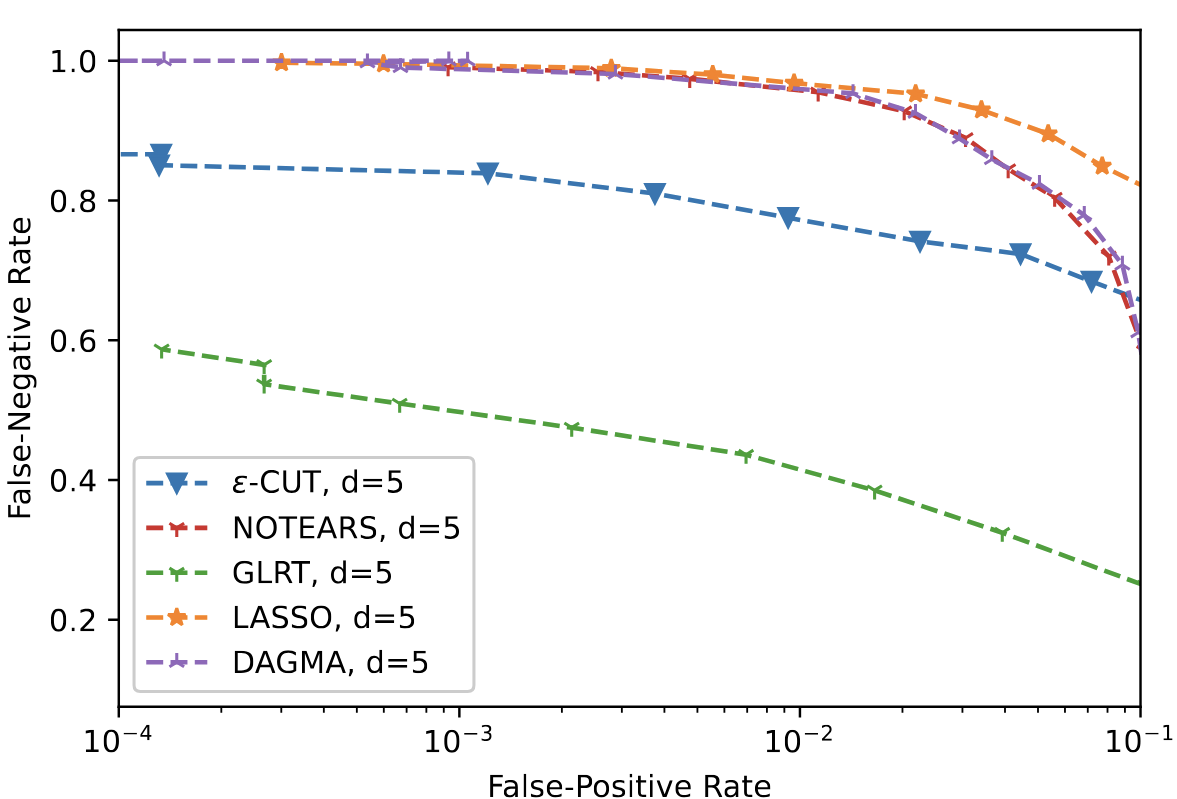}
    \caption{Performance curves for $d=5$.}
    \label{fig::p5a2}
    \end{subfigure}
    \begin{subfigure}{0.4\textwidth}
    \centering
    \includegraphics[scale=.35]{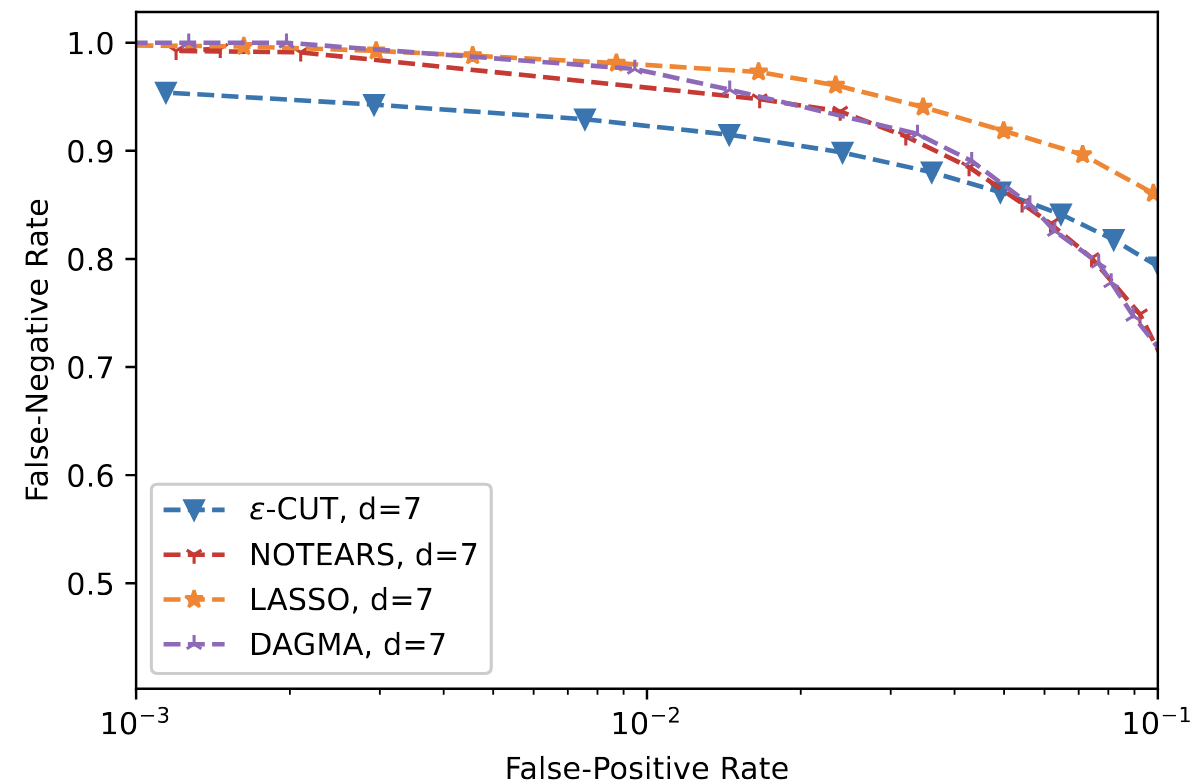}
    \caption{Performance curves for $d=7$.}
    \label{fig::p7a2}
    \end{subfigure}
    \begin{subfigure}{0.4\textwidth}
    \centering
    \includegraphics[scale=.35]{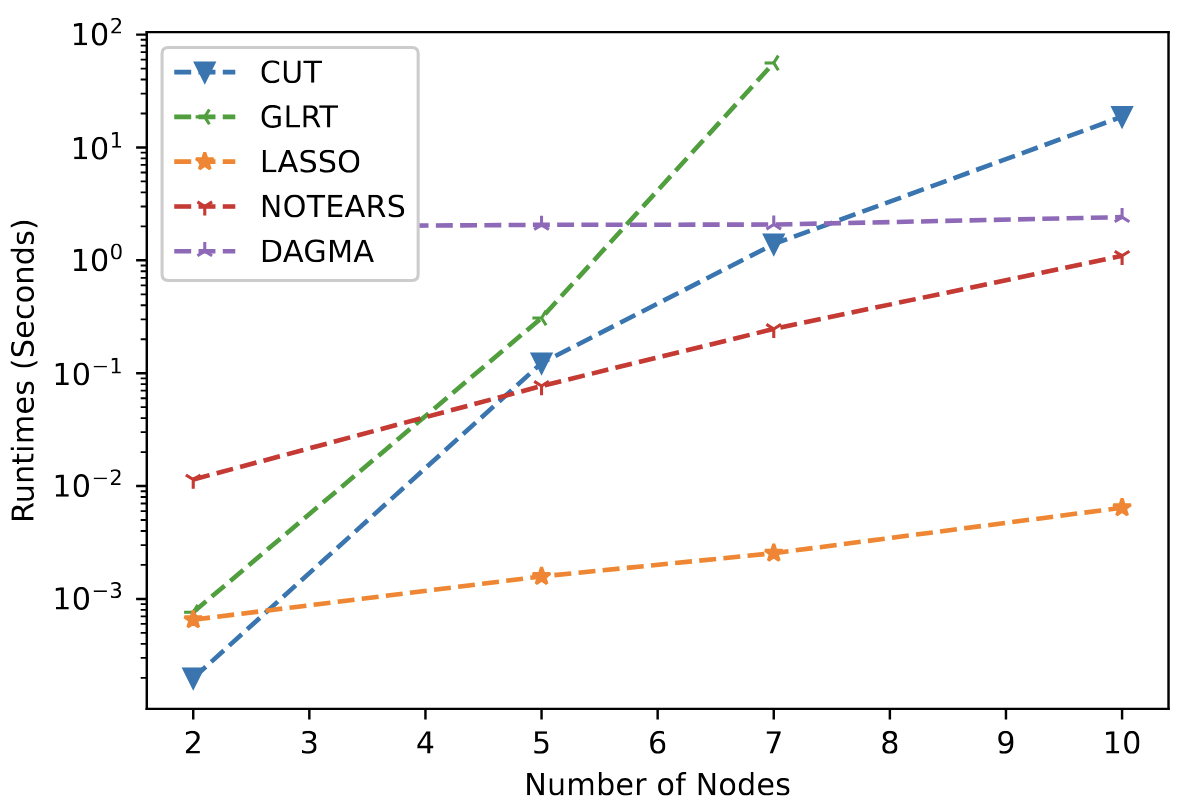}
    \caption{Log run-times on the various algorithms described in Section \ref{Sec::prioralgorithms}. Each curve is averaged over 100 runs.}
    \label{fig:runtimes}
    \end{subfigure}
    \begin{subfigure}{0.4\textwidth}
    \centering
    \includegraphics[scale=.35]{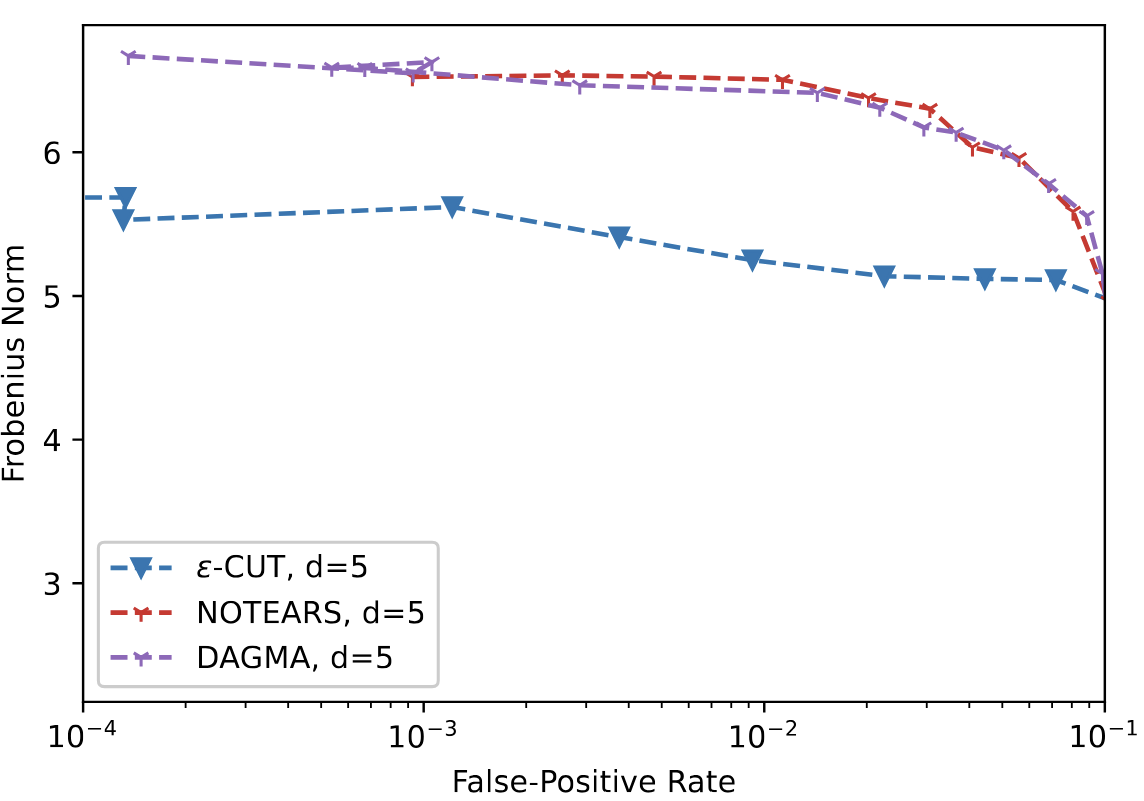}
    \caption{Frobenius norm between estimated DAG and true DAG for $d=5$ when $a\in\{1,2\}$.}
    \label{fig::MSE}
    \end{subfigure}
    \caption{Numerical results of the performance of various detectors. For all the cases, $\sigma^2=1$ and $n=10$ and the regularization constant $\lambda$ is varied for LASSO, NOTEARS, and DAGMA to control the error rates. The regularizer values used to generate the plots are given in the supplemental file: https://github.com/shaskajo/epsilonCUT. Performance curves are averaged over $1500$ iterations.}
    \label{fig:numerical}
    \vspace*{-0.2in}
\end{figure*}
We consider a numerical example to compare $\epsilon$-CUT to other popular algorithms that were summarized in Section \ref{Sec::prioralgorithms}. We examine some interesting phenomena and show the potential gains of $\epsilon$-CUT over the abovementioned algorithms.

\subsection{Definition of Graph Prior}
The number of DAGS grows super-exponentially in the number of vertices $d$ \cite{sprites_book_2017}, thus, generating all possible DAGS and uniform, random selection becomes computationally infeasible for increasing $d$. We use the following prior, $\pi_{\boldsymbol{A}}$ which is closely related to the priors used in \cite{HayekLower, NOTEARS}, and defined as follows:
\begin{enumerate}
    \item For a given number of vertices $d$, let $\boldsymbol{D}$ denote a $d\times d$ strictly lower triangular matrix. Then, the lower diagonal entries (support) are zeroed out randomly and independently. Then, for each remaining edge, the weight is selected uniformly at random from the set $[-5,-.5]\cup [.5, 5]$.
    \item Randomly relabel the vertices of $\boldsymbol{D}$. This operation is motivated by the fact that for any DAG, there exists a relabelling of the vertices such that the corresponding adjacency matrix is strictly lower triangular \cite{sprites_book_2017}.
\end{enumerate}

\subsection{Results}
Some numerical comparisons for a low number of samples ($n=10$) are given in Figure \ref{fig:numerical}. In Figure \ref{fig::p5a2} we compare the performance between $\epsilon-CUT$ and the previous state-of-the-art methods for $d=5$ vertices, and in Figure \ref{fig::p7a2} for $d=7$. In both cases, our algorithm is competitive with state-of-the-art methods (aside from the GLRT) in the low tolerance regime. In Figure \ref{fig::p5a2}, our method outperforms state-of-the-art methods by nearly $20\%$ in the low tolerance regime. While the GLRT significantly outperforms the other methods, it can be seen from Table \ref{tab:complexities} and Figure \ref{fig:runtimes} that the GLRT suffers from high computational complexity even for moderately sized graphs (this complexity is the reason that the GLRT is omitted from Figure \ref{fig::p7a2}, its complexity is too high). Moreover, in addition to $\epsilon-CUT$ offering competitive performance, we emphasize that, unlike LASSO, NOTEARS, and DAGMA, our method \textit{does not require hyperparameter tuning}, in addition to possessing \textit{finite-sample performance guarantees}.
\par
There is another interesting point to be made about the numerical results. First, from the definition of $\pi_{\boldsymbol{A}}$ and the description of the algorithms used, the edge error probabilities $P^+_{i,j}$ and $Q^-_{i,j}$ are the same for all pairs $(i,j)$. Hence, $\epsilon^+ = \sum_{i,j}w^+_{i,j}P^+_{i,j} = P^+_{1,2} \sum_{i,j}w^+_{i,j} = P^+_{1,2}$, and similarly for $\epsilon^-$. Hence, the curves figures \ref{fig::p5a2} and \ref{fig::p7a2} are also for the \textit{probability of error of individual edges with finite samples}. This is in contrast to most theoretical guarantees, which often deal with full graph recovery \cite{equivariance, NOTEARS,DAGMA}, or connected neighborhoods in the finite sample setting \cite{lasso}.

\subsection{Runtime Comparisons}
We compare the runtimes for the various algorithms in Figure \ref{fig:runtimes}. Surprisingly, $\epsilon-CUT$ is competitive with continuous optimization formulations (DAGMA and NOTEARS) even for moderately sized graphs. Of course, as the number of vertices increases, $\epsilon-CUT$ begins to fall behind DAGMA and NOTEARS. This is to be expected since $\epsilon-CUT$ requires an exhaustive search procedure for each pair of vertices to satisfy the false-positive constraint, leading to increased runtimes.
We reiterate that $\epsilon-CUT$ does not require hyperparameter tuning to satisfy the false-positive constraint. 

\subsection{Recovering Directions From Support}
Throughout this paper, we have focused only on the support recovery of a directed acyclic graph as this is the most challenging part of the graph identification problem. Note that the maximum likelihood estimate of the underlying DAG structure without any sparsity constraints can be found by considering only complete DAGs, i.e., DAGs with the maximum number of edges (Proposition 3.2 in \cite{StriederDrtonconfidence}); thus one must consciously enforce sparsity. However, if one does not care to enforce sparsity during the maximum likelihood estimation, then the maximum likelihood estimate may be found efficiently \cite{Chen_2019_ordering,StriederDrtonconfidence}.  Thus, the maximum likelihood estimate can be used to estimate link direction and the associated weights.
\par
We first use $\epsilon-CUT$ to estimate the support $\hat{\boldsymbol{\chi}}$. Then, we then compute the \textit{unrestricted maximum likelihood estimate},


\begin{equation}
    \min_{A \in \mathcal{M}} \frac{nd}{2}\log(2\pi\sigma^2) + \frac{1}{\sigma^2}\text{Trace}((\boldsymbol{I}-\boldsymbol{A})^\top (\boldsymbol{I}-\boldsymbol{A}\hat{\Sigma})
\end{equation}
where $\mathcal{M}$ is the set of all DAGs, and $\hat{\Sigma}$ is the empirical covariance matrix. Denoting the maximum likelihood estimate as $\hat{\boldsymbol{A}}_{ML}$, the final estimate is $\hat{\boldsymbol{A}} = \hat{\boldsymbol{A}}_{ML} \circ \hat{\boldsymbol{\chi}}$ which contains the edge weights and directions of the causal graph while inheriting
the theoretical guarantees on the sparsity provided by $\epsilon-CUT$. In Fig. \ref{fig::MSE} we plot the Frobenius norm between the true matrix $\boldsymbol{A}$ and the estimated matrix $\hat{\boldsymbol{A}}$ obtained from the procedure described above. In addition, we compare the performance of this procedure with DAGMA and NOTEARS, both of which output a DAG (as opposed to only the support). Except in the regime for larger false alarm rates, $\epsilon$-cut yields superior performance over NOTEARS and DAGMA for full DAG recovery.  Thus,
our method is competitive with both continuous optimization formulations, despite focusing only on support recovery. 




\section{Conclusions}\label{Sec::concl}
We have introduced a framework for \textit{Neyman-Pearson causal discovery}. Often, one needs to control one type of edge error, so our framework seeks to minimize one error rate subject to the other being below a user-specified tolerance level. This approach allowed us to derive the optimal detector for this problem. We derived finite sample performance bounds for this optimal detector. Surprisingly,  there is a strong performance gap between current methods and the optimal detector. Unfortunately, the optimal detector becomes computationally infeasible even for modestly sized graphs, which motivated us to derive $\epsilon-CUT$, which scales more favorably with graph size. We showed that $\epsilon-CUT$ provides finite-sample guarantees, and hence is feasible in the Neyman-Pearson causal discovery framework. In addition, we showed that $\epsilon-CUT$ is competitive with state-of-the-art methods \textit{without the need to tune regularizers or hyper-parameters while offering performance guarantees in the finite sample regime.} Finally, we considered the ability to recover both directions and edge weights, since $\epsilon-CUT$ only recovers support, and show that our method in tandem with unrestricted maximum likelihood estimations is better than state-of-the-art methods that output the full DAG for small false-alarm rate constraints.

\appendix
\subsection{Definitions}\label{AppenDefinitions}
\begin{definition}
    The \underline{R\'enyi divergence of order $\lambda$} between two probability measures $P$ and $Q$ is given as
    \begin{equation}
        D_{\lambda}(P||Q) \doteq \frac{1}{\lambda-1} \log \int_\mathcal{X} \Big( \frac{dP}{dQ} \Big)^{\lambda}dQ,
    \end{equation}
    where $\frac{dP}{dQ}$ is the Radon-Nikodym derivative of P with respect to Q. 
\end{definition}
\begin{definition}
    For any two probability measures $P$ and $Q$, let
    \begin{align}
        &D_{\lambda}'(P||Q) \doteq \int_\mathcal{X} F_{\lambda}(x;P,Q) \log \frac{dP}{dQ}, \\
         &D_{\lambda}''(P||Q) \doteq \int_\mathcal{X} F_{\lambda}(x;P,Q) \bigg(\log \frac{dP}{dQ} \bigg)^2 -\Big( D'_{\lambda}(P||Q) \Big)^2,
       \\
&\mbox{where}    \;\;\;      F_\lambda(x;P,Q)  \doteq \frac{f(x)^\lambda g(x)^{1-\lambda}}{\int_{\mathcal{X}} f(x)^\lambda g(x)^{1-\lambda} d\mu}, \; \; \;  \mbox{for} \;\;\lambda\in[0,1].
    \end{align}
\end{definition}
\subsection{Proof of Proposition \ref{upperbound}}\label{ProofofUpper}
\begin{proof}
for any $i,j$, we have
\begin{align}
    Q^-_{i,j} &\leq Q_{i,j}\Big(\frac{dP_{i,j}}{dQ_{i,j}} \geq \frac{w_{i,j}^-}{w_{i,j}^+}\gamma\Big) \\
    &= \int_{\mathcal{X}} \mathbbm{1} \Big\{ \frac{dP_{i,j}}{dQ_{i,j}} \geq \frac{w_{i,j}^-}{w_{i,j}^+}\gamma \Big\} dQ_{i,j}
    \\
    &\overset{(a)}{\leq}  \int_{\mathcal{X}} \Big( \frac{w_{i,j}^+}{w_{i,j}^1} \frac{1}{\gamma}\frac{dP_{i,j}}{dQ_{i,j}} \Big)^\lambda dQ_{i,j}
    \\
    &= \Big(\frac{w_{i,j}^+}{w_{i,j}^+}\Big)^\lambda\frac{1}{\gamma^\lambda} \int_{\mathcal{X}} \Big( \frac{dP_{i,j}}{dQ_{i,j}} \Big)^{\lambda} dQ_{i,j},
\end{align}
where $(a)$ holds since $\mathbbm{1}\{a\geq\gamma\}\leq \Big(\frac{a}{\gamma}\Big)^\lambda$ for any $a,\gamma,\lambda > 0$. Since, the above holds for any $\lambda>0$. Multiplying by $w^-_{i,j}$ yields
\begin{equation}
    w^-_{i,j}Q^-_{i,j}\leq (w^-_{i,j})^{1-\lambda} (w^+_{i,j})^{\lambda} \frac{1}{\gamma^{\lambda}} e^{-(1-\lambda)D_{\lambda}(P_{i,j}||Q_{i,j})}.
\end{equation}
Summing over $i,j$ completes the proof.
\end{proof}
\subsection{Proof of Proposition \ref{converse}}\label{ProofofConverse}
\begin{proof}
    It is relatively easy to show that for any pair $(i,j)$,
    \begin{align}
        Q_{i,j}(X_1,&..,X_n) = \exp\{-(1-\lambda)D_{\lambda}(P_{i,j}||Q_{i,j})\\
        & {-\lambda \log \frac{dP_{i,j}}{dQ_{i,j}} \}F_\lambda(X_1,...,X_n;P_{i,j},Q_{i,j})}, \\
        P_{i,j}(X_1,&..,X_n) =  \exp\{-(1-\lambda)D_{\lambda}(P_{i,j}||Q_{i,j})\\
        & +(1-\lambda) \log \frac{dP_{i,j}}{dQ_{i,j}}\} F_{\lambda}(X_1,...,X_n;P_{i,j},Q_{i,j}).
    \end{align}
    Then, consider the following sets
    \begin{align}
    \begin{split}
        \mathcal{X}_{i,j,\lambda} &= \bigg\{X_1,...,X_n : |\log \frac{dP_{i,j}}{dQ_{i,j}} - D'_{\lambda}(P_{i,j}||Q_{i,j})| \leq 
        \\ & \qquad \qquad \qquad \qquad \sqrt{ 2D_{\lambda}''(P_{i,j}||Q_{i,j})} \bigg\}, 
        \end{split}
        \\
        \mathcal{X}^0_{i,j} &= \bigg\{X_1,...,X_n : \hat{\boldsymbol{\chi}}_{i,j} = 0 \bigg\}, \\
        \mathcal{X}^1_{i,j} &= \bigg\{X_1,...,X_n : \hat{\boldsymbol{\chi}}_{i,j} = 1 \bigg\}.
    \end{align}
    Hence, for any $X_1,...,X_n \in \mathcal{X}_{i,j,\lambda}$, we have that
    \begin{align}
    \begin{split}
        &Q_{i,j}(X_1,..,X_n) \geq
        \\
        &\exp\bigg\{-(1-\lambda)D_{\lambda}(P_{i,j}||Q_{i,j}) -\lambda D'_{\lambda}(P_{i,j}||Q_{i,j}) 
        \\
        &\qquad  -\lambda \sqrt{2D_{\lambda}''(P_{i,j}||Q_{i,j})} \bigg\} F_\lambda(X_1,...,X_n;P_{i,j},Q_{i,j}),\label{lowerQ} 
        \end{split}
        \\
        \begin{split}
        &P_{i,j}(X_1,..,X_n) \geq 
        \\
        &\exp\bigg\{-(1-\lambda)D_{\lambda}(P_{i,j}||Q_{i,j}) +(1-\lambda)  D'_{\lambda}(P_{i,j}||Q_{i,j}) 
        \\
        & -(1-\lambda )\sqrt{2D_{\lambda}''(P_{i,j}||Q_{i,j})}\bigg\} F_\lambda(X_1,...,X_n;P_{i,j},Q_{i,j}). \label{lowerP}
        \end{split}
    \end{align}
    Notice that $D'_{\lambda}(P_{i,j}||Q_{i,j})$ is actually the mean of $\log \frac{dP_{i,j}}{dQ_{i,j}}$ with respect to the distribution $F_\lambda$, and $D''_{\lambda}(P_{i,j}||Q_{i,j})$ is its variance. Then, $\mathcal{X}_{i,j,\lambda}$ is the event that the log-likelihood ratio is within $\sqrt{2}$ standard deviations of its mean, and so from Chebychev's inequality,
    \begin{equation}
        \int_{\mathcal{X}_{i,j,\lambda}} F_\lambda(X_1,...,X_n;P_{i,j},Q_{i,j}) \geq \frac{1}{2}.
    \end{equation}
    Then, from the union bound, we have that
    \begin{equation}\label{chebybound}
        \begin{aligned}
            &\int_{\mathcal{X}_{i,j,\lambda\cap \mathcal{X}^0_{i,j}}} F_\lambda(X_1,...,X_n;P_{i,j},Q_{i,j}) 
            \\
            & \qquad \qquad + \int_{\mathcal{X}_{i,j,\lambda\cap \mathcal{X}^1_{i,j}}} F_\lambda(X_1,...,X_n;P_{i,j},Q_{i,j}) \geq \frac{1}{2}.
        \end{aligned}
    \end{equation}
    From \eqref{lowerP}, we have that
    \begin{align}
    &\int_{\mathcal{X}_{i,j,\lambda\cap \mathcal{X}^1_{i,j}}} F_\lambda(X_1,...,X_n;P_{i,j},Q_{i,j})
    \\
    \begin{split}
    &\leq \exp\bigg\{(1-\lambda)D_{\lambda}(P_{i,j}||Q_{i,j}) -(1-\lambda)  D'_{\lambda}(P_{i,j}||Q_{i,j}) 
    \\
    & +(1-\lambda )\sqrt{2D_{\lambda}''(P_{i,j}||Q_{i,j})}\bigg\} \int_{\mathcal{X}_{i,j,\lambda\cap \mathcal{X}^1_{i,j}}} P_{i,j}(X_1,...,X_n)
    \end{split}
    \\
    \begin{split}
    &\leq \exp\bigg\{(1-\lambda)D_{\lambda}(P_{i,j}||Q_{i,j}) -(1-\lambda)  D'_{\lambda}(P_{i,j}||Q_{i,j})
    \\
    & \qquad \qquad \qquad \qquad +(1-\lambda )\sqrt{2D_{\lambda}''(P_{i,j}||Q_{i,j})}\bigg\}  P^+_{i,j}
    \end{split}
    \end{align}
    and so
    \begin{equation}\label{lowerF}
        \begin{aligned}
            &\int_{\mathcal{X}_{i,j,\lambda\cap \mathcal{X}^0_{i,j}}} F_\lambda(X_1,...,X_n;P_{i,j},Q_{i,j}) \geq \frac{1}{2} - 
            \\
            &\exp\bigg\{(1-\lambda)D_{\lambda}(P_{i,j}||Q_{i,j}) -(1-\lambda)  D'_{\lambda}(P_{i,j}||Q_{i,j})
            \\
            &\qquad \qquad +(1-\lambda )\sqrt{2D_{\lambda}''(P_{i,j}||Q_{i,j})}\bigg\}P_{i,j}^+.
        \end{aligned}
    \end{equation}
    Then, \eqref{lowerQ} gives
    \begin{equation}
        \begin{aligned}
            &Q_{i,j}^- \geq  \exp\bigg\{-(1-\lambda)D_{\lambda}(P_{i,j}||Q_{i,j})
            \\
            &\qquad \qquad -\lambda D'_{\lambda}(P_{i,j}||Q_{i,j})-\lambda \sqrt{2D_{\lambda}''(P_{i,j}||Q_{i,j})} \exp\bigg\}
            \\&\qquad \qquad \qquad \int_{\mathcal{X}_{i,j,\lambda\cap \mathcal{X}^0_{i,j}}} F_\lambda(X_1,...,X_n;P_{i,j},Q_{i,j}),
        \end{aligned}
    \end{equation}
    together with \eqref{lowerF} we have
    \begin{equation}
        \begin{aligned}
             &Q_{i,j}^- \geq  \frac{1}{2}\exp\bigg\{-(1-\lambda)D_{\lambda}(P_{i,j}||Q_{i,j})-\lambda D'_{\lambda}(P_{i,j}||Q_{i,j})
             \\
             & \qquad-\lambda \sqrt{2D_{\lambda}''(P_{i,j}||Q_{i,j})} \bigg\} - \exp\bigg\{-D'_{\lambda}(P_{i,j}||Q_{i,j})
             \\
             & \qquad \qquad \qquad \qquad +(1-2\lambda)\sqrt{2D_{\lambda}''(P_{i,j}||Q_{i,j})}\bigg\} P_{i,j}^+.
        \end{aligned}
    \end{equation}
    Multiplying both sides by $w^-_{i,j}$ and summing over all $(i,j)$ pairs gives us
    \begin{equation}
        \begin{aligned}
        &\sum_{i,j} w^-_{i,j}Q_{i,j}^- \geq \frac{1}{2}\sum_{i,j} w^-_{i,j}\exp\bigg\{-(1-\lambda)D_{\lambda}(P_{i,j}||Q_{i,j})
        \\
        & -\lambda D'_{\lambda}(P_{i,j}||Q_{i,j})-\lambda \sqrt{2D_{\lambda}''(P_{i,j}||Q_{i,j})} \bigg\}  - \sum_{i,j} w^+_{i,j} \frac{w^-_{i,j}}{w^+_{i,j}}
             \\
             &\exp\bigg\{-D'_{\lambda}(P_{i,j}||Q_{i,j})+(1-2\lambda)\sqrt{2D_{\lambda}''(P_{i,j}||Q_{i,j})}\bigg\} P_{i,j}^+,
        \end{aligned}
    \end{equation}
    which yields 
     \begin{equation}
        \begin{aligned}
        &\sum_{i,j} w^-_{i,j}Q_{i,j}^- \geq \frac{1}{2}\sum_{i,j} w^-_{i,j}\exp\bigg\{-(1-\lambda)D_{\lambda}(P_{i,j}||Q_{i,j})
        \\
        & -\lambda D'_{\lambda}(P_{i,j}||Q_{i,j})-\lambda \sqrt{2D_{\lambda}''(P_{i,j}||Q_{i,j})} \bigg\}-  \max_{i,j}\bigg\{ \frac{w^-_{i,j}}{w^+_{i,j}}
             \\
             &\exp\bigg\{-D'_{\lambda}(P_{i,j}||Q_{i,j})+(1-2)\lambda \sqrt{2D_{\lambda}''(P_{i,j}||Q_{i,j})}\bigg\}\bigg\} 
             \\
             & \qquad \qquad \qquad \qquad \qquad \qquad \qquad \qquad \qquad  \sum_{i,j} w^+_{i,j} P_{i,j}^+.
        \end{aligned}
    \end{equation}
    Noticing that $\sum_{i,j} w^+_{i,j} P_{i,j}^+ = \epsilon^+ \leq \epsilon$ completes the proof.
\end{proof}

\subsection{Proof of Theorem \ref{optimalrule}.}\label{proof::optimal}
We first present the proof of Theorem \ref{optimalrule}, which closely resembles that of the Neyman-Pearson Lemma (Proposition II.D.1 in \cite{poor2013introduction}). Nonetheless, we include it here for completeness.
\begin{proof}
We proceed first by showing existence. Let $\gamma \geq 0$ be the largest $\gamma$ such that
\begin{equation}
    \sum_{i,j} w^+_{i,j} P_{i,j} \bigg( \frac{dP_{i,j}}{dQ_{i,j}} < \frac{w_{i,j}^-}{w_{i,j}^+}\gamma\bigg) \leq \epsilon.
\end{equation}
Then, if the inequality is strict, select $\eta$ to be
\begin{equation}
    \eta = \frac{\epsilon - \sum_{i,j} w^+_{i,j} P_{i,j} \bigg( \frac{dP_{i,j}}{dQ_{i,j}} < \frac{w_{i,j}^-}{w_{i,j}^+}\gamma\bigg)}{\sum_{i,j} w^+_{i,j} P_{i,j} \bigg( \frac{dP_{i,j}}{dQ_{i,j}} = \frac{w_{i,j}^-}{w_{i,j}^+}\gamma\bigg)}
\end{equation}
otherwise choose $\eta$ arbitrarily. Hence, we have that
\begin{align}
\begin{split}
    \epsilon^+ &=  \sum_{i,j} w^+_{i,j}P_{i,j}^+ = \sum_{i,j} w^+_{i,j}\Bigg( P_{i,j} \bigg( \frac{dP_{i,j}}{dQ_{i,j}} < \frac{w_{i,j}^-}{w_{i,j}^+}\gamma\bigg)
    \\
    & \qquad \qquad \qquad \qquad \qquad + \eta P_{i,j} \bigg( \frac{dP_{i,j}}{dQ_{i,j}} = \frac{w_{i,j}^-}{w_{i,j}^+}\gamma\bigg)\Bigg)
    \end{split}
    \\
    \begin{split}
    &= \sum_{i,j} w_{i,j}^+ P_{i,j} \bigg( \frac{dP_{i,j}}{dQ_{i,j}} < \frac{w_{i,j}^-}{w_{i,j}^+}\gamma\bigg) 
    \\
    &\qquad \qquad + \eta \sum_{i,j} w_{i,j}^+ P_{i,j} \bigg( \frac{dP_{i,j}}{dQ_{i,j}} = \frac{w_{i,j}^-}{w_{i,j}^+}\gamma\bigg) = \epsilon.
    \end{split}
\end{align}
We next show that threshold rules are optimal. We use the Lagrange multiplier $\lambda_0 > 0$ (if the optimal estimator is such that $\lambda_0=0$ then the constraint is inactive, and so the optimal estimator should always declare $\hat{\boldsymbol{\chi}}_{i,j}=1$) and seek to minimize $\epsilon^- + \lambda_0 \epsilon^+$. Then, we have
\begin{align}
\begin{split}
    &\epsilon^- + \lambda_0 \epsilon^+ = \sum_{i,j} w^-_{i,j} \mathbb{P}(\hat{\boldsymbol{\chi}}_{i,j}=0|\boldsymbol{\chi}_{i,j}=1)
    \\
    & \qquad \qquad + \lambda_0 \sum_{i,j} w^+_{i,j} \mathbb{P}(\hat{\boldsymbol{\chi}}_{i,j}=1|\boldsymbol{\chi}_{i,j}=0)
    \end{split}
    \\
    \begin{split}
    = &\sum_{i,j} \int_{X_1,...,X_n} w_{i,j}^-\mathbb{P}(\hat{\boldsymbol{\chi}}_{i,j}=0| X_1,...,X_n) Q_{i,j}(X_1,..,X_n)
    \\
    & \qquad \qquad +\lambda_0 w_{i,j}^+ \mathbb{P}(\hat{\boldsymbol{\chi}}_{i,j}=1| X_1,...,X_n) P_{i,j} (X_1,..,X_n).
    \end{split}
\end{align}
Then, to minimize $\epsilon^- + \lambda_0 \epsilon^+$, the estimator $\hat{\boldsymbol{\chi}}$ should be such that $\hat{\boldsymbol{\chi}}_{i,j}=0$ if $w^+_{i,j}Q_{i,j} \leq \lambda_0 w_{i,j}^-P_{i,j}(X_1,...,X_n)$ and $\hat{\boldsymbol{\chi}}_{i,j}=1$ otherwise. Since threshold rules are optimal, it remains to be seen how to select the threshold $\gamma$. Observe that the probabilities
$P_{i,j} \bigg( \frac{dP_{i,j}}{dQ_{i,j}} < \frac{w_{i,j}^-}{w_{i,j}^+}\gamma\bigg)$ are increasing in $\gamma$ for all pairs $(i,j)$, and that the probabilities $Q_{i,j} \bigg( \frac{dP_{i,j}}{dQ_{i,j}} > \frac{w_{i,j}^-}{w_{i,j}^+}\gamma\bigg)$ are decreasing in $\gamma$. Hence, $\epsilon^+$ is increasing in $\gamma$, and $\epsilon^-$ is decreasing in $\gamma$. Then, to minimize $\epsilon^-$, $\gamma$ should be chosen to make $\epsilon^+$ as large as possible. 
\end{proof}
\subsection{Proof of Theorem \ref{satisfyconstraint}.} \label{proof::eCUT}
\begin{proof}
First, recall that $\epsilon^+ = \sum_{i,j} w^+_{i,j}P_{i,j}^+$, and so if it holds that for all $i,j$, $P^+_{i,j} \leq \epsilon$, we get $\epsilon^+ = \sum_{i,j} w^+_{i,j}P_{i,j}^+ \leq \epsilon \sum_{i,j} w^+_{i,j} = \epsilon$. Then, it suffices to show that for all $i,j$, $P^+_{i,j} \leq \epsilon$. We begin by writing
\begin{equation}
    \begin{aligned}
&P^+_{i,j} = \int_{\boldsymbol{A}}\mathbb{P}_{\boldsymbol{A}}(\hat{\boldsymbol{\chi}}_{i,j}=1|\boldsymbol{\chi}_{i,j}=0) \mathbb{P}(\boldsymbol{A}|\boldsymbol{\chi}_{i,j}=0) 
\\
&\overset{(a)}{=} \int_{\boldsymbol{A}}\mathbb{P}_{\boldsymbol{A}}(\hat{\boldsymbol{\chi}}_{i,j}=1) \mathbb{P}(\boldsymbol{A}|\boldsymbol{\chi}_{i,j}=0) 
\\ &= \int_{\boldsymbol{A}}\int_{\boldsymbol{X}^{\setminus i,j}}\mathbb{P}_{\boldsymbol{A}}(\hat{\boldsymbol{\chi}}_{i,j}=1|\boldsymbol{X}^{\setminus i,j}) \mathbb{P}_{\boldsymbol{A}}(\boldsymbol{X}^{\setminus i,j}) \mathbb{P}(\boldsymbol{A}|\boldsymbol{\chi}_{i,j}=0)
    \end{aligned}
\end{equation}
where $\boldsymbol{X}^{\setminus i,j}$ denotes the set of all measurements except those from the $i$th and $j$th vertex. That is, $\boldsymbol{X}^{\setminus i,j} = \{X^{\setminus i, j}_k\}_{k=1}^n$ where $X^{\setminus i, j}_k = [X_k(1),...,X_k(i-1),X_k(i+1),...,X_k(j-1), X_k(j+1),...,X_k(d)]^\top$. $(a)$ holds since $\boldsymbol{\chi}_{i,j}$ is a deterministic function of $\boldsymbol{A}$ which gives us
\begin{align}
    &\mathbb{P}_{\boldsymbol{A}}(\hat{\boldsymbol{\chi}}_{i,j}=1|\boldsymbol{\chi}_{i,j}=0) = \frac{\mathbb{P}(\hat{\boldsymbol{\chi}}_{i,j}=1, \boldsymbol{\chi}_{i,j}=0,\boldsymbol{A})}{\mathbb{P}(\boldsymbol{A},\boldsymbol{\chi}_{i,j}=0)}
    \\
    &= \frac{\mathbb{P}(\boldsymbol{\chi}_{i,j}=0| \hat{\boldsymbol{\chi}}_{i,j}=1, \boldsymbol{A})\mathbb{P}_{\boldsymbol{A}}(\hat{\boldsymbol{\chi}}_{i,j}=1)\pi_{\boldsymbol{A}}}{\mathbb{P}(\boldsymbol{\chi}_{i,j}=0|\boldsymbol{A})\pi_{\boldsymbol{A}}}
    \\
     &= \frac{\mathbb{P}(\boldsymbol{\chi}_{i,j}=0| \boldsymbol{A})\mathbb{P}_{\boldsymbol{A}}(\hat{\boldsymbol{\chi}}_{i,j}=1)\pi_{\boldsymbol{A}}}{\mathbb{P}(\boldsymbol{\chi}_{i,j}=0|\boldsymbol{A})\pi_{\boldsymbol{A}}}
     = \mathbb{P}_{\boldsymbol{A}}(\hat{\boldsymbol{\chi}}_{i,j}=1).
\end{align}
Now, if we can say that for any $\boldsymbol{A}$ and $\boldsymbol{X}^{\setminus i,j}$ we have $\mathbb{P}_{\boldsymbol{A}}(\hat{\boldsymbol{\chi}}_{i,j}=1|\boldsymbol{X}^{\setminus i,j}) \leq \epsilon$,
then the proof is complete since that would give us 
\begin{equation}
    \begin{aligned}
\int_{\boldsymbol{A}}\int_{\boldsymbol{X}^{\setminus i,j}}\mathbb{P}_{\boldsymbol{A}}(\hat{\boldsymbol{\chi}}_{i,j}=1|\boldsymbol{X}^{\setminus i,j}) \mathbb{P}_{\boldsymbol{A}}(\boldsymbol{X}^{\setminus i,j}) \mathbb{P}(\boldsymbol{A}|\boldsymbol{\chi}_{i,j}=0)
    \\
    \leq \epsilon \int_{\boldsymbol{A}}\int_{\boldsymbol{X}^{\setminus i,j}} \mathbb{P}_{\boldsymbol{A}}(\boldsymbol{X}^{\setminus i,j}) \mathbb{P}(\boldsymbol{A}|\boldsymbol{\chi}_{i,j}=0) = \epsilon.
    \end{aligned}
\end{equation}

We focus on the event $\hat{\boldsymbol{\chi}}_{i,j}=1$. 
Notice that our algorithm only declares an edge between the vertices $i$ and $j$ if for \textit{all} sets of potential parents, which we denote as $\hat{\mathcal{Z}}(i)$ and $\hat{\mathcal{Z}}(j)$, respectively, we have that $|\hat{\sigma}^{2}_i - \hat{\sigma}^{2}_j| >  \tau_{\hat{\mathcal{Z}}(i),\hat{\mathcal{Z}}(j)}$,
where
\begin{align}
    \hat{\sigma}^{2}_i &= \sum_{k=1}^n \big(X_k(i) - \hat{\alpha}^\top \hat{Z}_k(i)\big)^2, \\
    \hat{\sigma}^{2}_j &=\sum_{k=1}^n \big(X_k(j) - \hat{\beta}^\top \hat{Z}_k(j)\big)^2,
\end{align}
and the vectors $\hat{Z}_k(i)$ and $\hat{Z}_k(i)$ contain the $k$th measurements of the vertices in $\hat{\mathcal{Z}}(i)$ and $\hat{\mathcal{Z}}(j)$, respectively. The vectors $\hat{\alpha}$ and $\hat{\beta}$ are the resulting coefficient vectors obtained by performing least squares of vertices $i$ and $j$ on the potential parent sets $\hat{\mathcal{Z}}(i)$ and $\hat{\mathcal{Z}}(j)$, respectively. $\tau_{\hat{\mathcal{Z}}(i),\hat{\mathcal{Z}}(j)}$ is chosen so that
\begin{equation}
    \begin{aligned}
2&-F_{n-p}(\tau'_{\hat{\mathcal{Z}}(i),\hat{\mathcal{Z}}(j)} + n- p) 
+ F_{n-p}(-\tau'_{\hat{\mathcal{Z}}(i),\hat{\mathcal{Z}}(j)} + n- p)
    \\
    & -F_{n-q}(\tau'_{\hat{\mathcal{Z}}(i),\hat{\mathcal{Z}}(j)} + n- q) 
    + F_{n-q}(-\tau'_{\hat{\mathcal{Z}}(i),\hat{\mathcal{Z}}(j)} + n- q)
    \\
    &=\epsilon,
    \end{aligned}
\end{equation}
where $\tau'_{\hat{\mathcal{Z}}(i),\hat{\mathcal{Z}}(j)} = (\tau_{\hat{\mathcal{Z}}(i),\hat{\mathcal{Z}}(j)}- |q-p|\sigma^2)/2\sigma^2$, $p = |\hat{\mathcal{Z}}(i)|$ and $q=|\hat{\mathcal{Z}}|(j)$ and $F_l$ is the cdf of the chi-squared distributions with $l$-degrees of freedom.

Then, we have that
\begin{align}
    &\mathbb{P}_{\boldsymbol{A}}(\hat{\boldsymbol{\chi}}_{i,j}=1|\boldsymbol{X}^{\setminus i,j}) 
    \\&= \mathbb{P}_{\boldsymbol{A}}\bigg(\bigcap_{\hat{\mathcal{Z}}(i),\hat{\mathcal{Z}}(j)}|\hat{\sigma}^{2}_i - \hat{\sigma}^{2}_j| >  \tau_{\hat{\mathcal{Z}}(i),\hat{\mathcal{Z}}(j)}|\boldsymbol{X}^{\setminus i,j}\bigg)
    \\
    &\leq \mathbb{P}_{\boldsymbol{A}}\big(|\hat{\sigma}^{*2}_i - \hat{\sigma}^{*2}_j| >  \tau_{\mathcal{Z}(i),\mathcal{Z}(j)}|\boldsymbol{X}^{\setminus i,j}\big) 
    \\
    \begin{split}
        &\overset{(a)}{=} \mathbb{P}_{\boldsymbol{A}}\big(|\hat{\sigma}^{*2}_i - \hat{\sigma}^{*2}_j - (n-p)\sigma^2 + (n-q)\sigma^2 + (q-p)\sigma^2| 
        \\
        &\qquad \qquad \qquad \qquad \qquad \qquad  \qquad \qquad  >  \tau_{\mathcal{Z}(i),\mathcal{Z}(j)}|\boldsymbol{X}^{\setminus i,j}\big)
    \end{split}
    \\
    \begin{split}
        &\overset{(b)}{\leq} \mathbb{P}_{\boldsymbol{A}}\big(|\hat{\sigma}^{*2}_i - (n-p)\sigma^2| + |\hat{\sigma}^{*2}_j - (n-q)\sigma^2|
        \\
        &\qquad \qquad \qquad \qquad \qquad + |(q-p)\sigma^2| >  \tau_{\mathcal{Z}(i),\mathcal{Z}(j)}|\boldsymbol{X}^{\setminus i,j}\big)
    \end{split}
    \\
    \begin{split}
        &= \mathbb{P}_{\boldsymbol{A}}\big(|\hat{\sigma}^{*2}_i - (n-p)\sigma^2| + |\hat{\sigma}^{*2}_j - (n-q)\sigma^2| 
        \\ &\qquad \qquad \qquad \qquad \qquad >  \tau_{\mathcal{Z}(i),\mathcal{Z}(j)} - |q-p|\sigma^2|\boldsymbol{X}^{\setminus i,j}\big)
    \end{split}
    \\
    \begin{split}
        &\overset{(c)}{\leq} \mathbb{P}_{\boldsymbol{A}}\big(|\hat{\sigma}^{*2}_i - (n-p)\sigma^2|>  \frac{\tau_{\mathcal{Z}(i),\mathcal{Z}(j)} - |q-p|\sigma^2}{2}
        \\
        & \qquad \bigcup |\hat{\sigma}^{*2}_j - (n-q)\sigma^2| >  \frac{\tau_{\mathcal{Z}(i),\mathcal{Z}(j)} - |q-p|\sigma^2}{2}|\boldsymbol{X}^{\setminus i,j}\big)
    \end{split}
    \\
    \begin{split}
        &\overset{(d)}{\leq} \mathbb{P}_{\boldsymbol{A}}\Big(\Big|\frac{\hat{\sigma}^{*2}_i}{\sigma^2} - (n-p)\Big|>  \tau'_{\mathcal{Z}(i),\mathcal{Z}(j)}|\boldsymbol{X}^{\setminus i,j}\Big)
        \\
        & + \mathbb{P}_{\boldsymbol{A}}\Big(\Big|\frac{\hat{\sigma}^{*2}_j}{\sigma^2} - (n-q)\Big| >  \tau'_{\mathcal{Z}(i),\mathcal{Z}(j)}|\boldsymbol{X}^{\setminus i,j}\Big). \label{unionineq}
    \end{split}
\end{align}
Where $(a)$ holds since 
\begin{equation}
    - (n-p)\sigma^2 + (n-q)\sigma^2 + (q-p)\sigma^2 = 0.
\end{equation}
$(b)$ holds since by the triangle inequality we have
\begin{equation}
    \begin{aligned}
        &|\hat{\sigma}^{*2}_i - \hat{\sigma}^{*2}_j - (n-p)\sigma^2 + (n-q)\sigma^2 + (q-p)\sigma^2| 
        \\
        &\leq |\hat{\sigma}^{*2}_i - (n-p)\sigma^2| + |\hat{\sigma}^{*2}_j - (n-q)\sigma^2| + |(q-p)\sigma^2|,
    \end{aligned}
\end{equation}
and so if the left side is greater than $\tau_{\mathcal{Z}(i),\mathcal{Z}(j)}$, this implies that right side is as well. Hence,
\begin{equation}
    \begin{aligned}
    &\bigg\{\tau_{\mathcal{Z}(i),\mathcal{Z}(j)} < 
    \\
        &|\hat{\sigma}^{*2}_i - \hat{\sigma}^{*2}_j - (n-p)\sigma^2 + (n-q)\sigma^2 + (q-p)\sigma^2| \bigg\}
        \\
        &  \subseteq \bigg\{\tau_{\mathcal{Z}(i),\mathcal{Z}(j)} < |\hat{\sigma}^{*2}_i - (n-p)\sigma^2| 
        \\
        & \qquad \qquad \qquad \qquad + |\hat{\sigma}^{*2}_j - (n-q)\sigma^2| + |(q-p)\sigma^2|\bigg\}.
    \end{aligned}
\end{equation}
$(c)$ holds due to the following fact: for any real numbers $a$, $b$, and $c$, 
\begin{equation}
    a + b > c \implies a > \frac{c}{2} \text{ or } b > \frac{c}{2},
\end{equation}
and $(d)$ holds by the union bound and the definition $\tau'_{\mathcal{Z}(i),\mathcal{Z}(j)} = (\tau_{\mathcal{Z}(i),\mathcal{Z}(j)}- |q-p|\sigma^2)/2\sigma^2$. For simplicity, we turn our attention to 
\begin{equation}
    \mathbb{P}_{\boldsymbol{A}}\Big(\Big|\frac{\hat{\sigma}^{*2}_i}{\sigma^2} - (n-p)\Big|>  \tau'_{\mathcal{Z}(i),\mathcal{Z}(j)}\Big|\boldsymbol{X}^{\setminus i,j}\Big). \label{boundoni}
\end{equation}
First, begin by noticing \eqref{boundoni} can be written as
\begin{equation}
    \begin{aligned}
        & \mathbb{P}_{\boldsymbol{A}}\Big(\frac{\hat{\sigma}^{*2}_i}{\sigma^2} - (n-p)>  \tau'_{\mathcal{Z}(i),\mathcal{Z}(j)}
        \\
        & \qquad \qquad \bigcup \frac{\hat{\sigma}^{*2}_i}{\sigma^2} - (n-p)<  -\tau'_{\mathcal{Z}(i),\mathcal{Z}(j)}\Big|\boldsymbol{X}^{\setminus i,j}\Big),
    \end{aligned}
\end{equation}
which is further upper bounded by
\begin{equation}
    \begin{aligned}
        & \mathbb{P}_{\boldsymbol{A}}\Big(\frac{\hat{\sigma}^{*2}_i}{\sigma^2} - (n-p)>  \tau'_{\mathcal{Z}(i),\mathcal{Z}(j)}\Big|\boldsymbol{X}^{\setminus i,j}\Big)
        \\
        &+ \mathbb{P}_{\boldsymbol{A}}\Big(\frac{\hat{\sigma}^{*2}_i}{\sigma^2} - (n-p)<  -\tau'_{\mathcal{Z}(i),\mathcal{Z}(j)}\Big|\boldsymbol{X}^{\setminus i,j}\Big).
    \end{aligned}
\end{equation}
We turn our attention to the first-term
\begin{equation}
    \mathbb{P}_{\boldsymbol{A}}\Big(\frac{\hat{\sigma}^{*2}_i}{\sigma^2} - (n-p)>  \tau'_{\mathcal{Z}(i),\mathcal{Z}(j)}\Big|\boldsymbol{X}^{\setminus i,j}\Big).\label{finalprob}
\end{equation}
From Lemma \ref{RSSChi}, $\frac{\hat{\sigma}^{*2}_i}{\sigma^2}$ follows a chi-squared distribution with $n-p$ degrees of freedom conditioned on $\boldsymbol{X}^{\setminus i,j}$. We have that \eqref{finalprob} is equal to
\begin{equation}
    1-F_{n-p}(\tau'_{\mathcal{Z}(i),\mathcal{Z}(j)} + n- p),
\end{equation}
where $F_{n-p}$ is the cumulative distribution function (cdf) of a chi-squared distribution with $n-p$ degrees of freedom. Following a similar argument, we also get that
\begin{equation}
     \begin{aligned}
&\mathbb{P}_{\boldsymbol{A}}\Big(\frac{\hat{\sigma}^{*2}_i}{\sigma^2} < - \tau'_{\mathcal{Z}(i),\mathcal{Z}(j)}+ n-p\Big|\boldsymbol{X}^{\setminus i,j}\Big) 
\\
&= F_{n-p}(-\tau'_{\mathcal{Z}(i),\mathcal{Z}(j)} + n- p).
     \end{aligned}
\end{equation}
Following the same argument, we have that
\begin{equation}
    \begin{aligned}
        &\mathbb{P}_{\boldsymbol{A}}\Big(\Big|\frac{\hat{\sigma}^{*2}_j}{\sigma^2} - (n-q)\Big| >  \tau'_{\mathcal{Z}(i),\mathcal{Z}(j)}\Big|\boldsymbol{X}^{\setminus i,j}\Big)
        \\
        &\leq 1-F_{n-q}(\tau'_{\mathcal{Z}(i),\mathcal{Z}(j)} + n- q) 
        \\ &\qquad \qquad + F_{n-q}(-\tau'_{\mathcal{Z}(i),\mathcal{Z}(j)} + n- q).
    \end{aligned}
\end{equation}
Hence, from \eqref{unionineq} we have
\begin{equation}
    \begin{aligned}
&\mathbb{P}_{\boldsymbol{A}}(\hat{\boldsymbol{\chi}}_{i,j}=1|\boldsymbol{X}^{\setminus i,j}) \leq
 2-F_{n-q}(\tau'_{\mathcal{Z}(i),\mathcal{Z}(j)}+n- q) 
\\
&- F_{n-p}(\tau'_{\mathcal{Z}(i),\mathcal{Z}(j)}+n- p) + F_{n-q}(-\tau'_{\mathcal{Z}(i),\mathcal{Z}(j)}+n- q)
    \\
    & \qquad \qquad \qquad \qquad \qquad \qquad + F_{n-p}(-\tau'_{\mathcal{Z}(i),\mathcal{Z}(j)}+n- p). \label{finalineq}
    \end{aligned}
\end{equation}
From the definition of $\tau_{\mathcal{Z}(i),\mathcal{Z}(j)}$, the right-hand side of \eqref{finalineq} is equal to $\epsilon$, which completes the proof.
\end{proof}

\subsection{Proof of Lemma \ref{chisquaredlemma}}\label{ProofChiSquared}
For notational simplicity and readability, we restate the lemma with simplified notation.
\begin{lemma}\label{RSSChi}
    Let $\{X_k\}_{k=1}^n$ be a \textbf{fixed} dataset of $p$-dimensional vectors. Let $\boldsymbol{X}$ be the matrix constructed by taking the $k$th row as $X_k^\top$. Let $y_k$ be defined as 
    \begin{equation*}
        y_k = X_k^\top \alpha + w_k,
    \end{equation*}
    where $\alpha$ is a vector of coefficients and $w_k$ is zero mean Gaussian noise with variance $\sigma^2$. Furthermore, let $\hat{\alpha} = (\boldsymbol{X}^\top\boldsymbol{X})^{-1}\boldsymbol{X}^\top Y$, i.e., the ordinary least squares estimate of $\alpha$, where $Y = [y_1, y_2, ..., y_n]^\top$. Then,
    \begin{equation*}
        \frac{1}{\sigma^2}\sum_{k=1}^n (y_k - X_k\hat{\alpha})^2,
    \end{equation*}
    follows a chi-squared distribution with $n-p$ degrees of freedom.
\end{lemma}

\begin{proof}[Proof of Lemma \ref{RSSChi}]
    We begin by defining the vector of residuals as
    \begin{align}
        \hat{W} = Y -\boldsymbol{X} \hat{\alpha}  
        = \boldsymbol{Q}Y = \boldsymbol{Q}( \boldsymbol{X}\alpha + W),
    \end{align}
    where $\boldsymbol{Q} = I-\boldsymbol{X}(\boldsymbol{X}^\top\boldsymbol{X})^{-1}\boldsymbol{X}^\top$ and $W = [w_1, w_2, ..., w_n]^\top$. Notice that $\boldsymbol{Q}\boldsymbol{X} = 0$, and that $\boldsymbol{Q}$ is a projection matrix. Hence, $\boldsymbol{Q}^2 = \boldsymbol{Q}$ which implies $\hat{W}$ is a Gaussian vector with zero mean and variance $\sigma^2 \boldsymbol{Q}$. Then, there exists some unitary matrix $\boldsymbol{U}$ that diagonalizes $\boldsymbol{Q}$  Moreover, notice that
    \begin{align}
        \text{Trace}(Q)  &\overset{(a)}{=} \text{Trace}(I-\boldsymbol{X}(\boldsymbol{X}^\top\boldsymbol{X})^{-1}\boldsymbol{X}^\top) 
         = n-p
    \end{align}
    where $(a)$ holds since the trace is invariant to cyclic shifts. Then, if we define $\hat{W}' = \hat{W}^\top \boldsymbol{U}$, we see that $\hat{W}'$ is a Gaussian vector with zero mean and covariance $\sigma^2\boldsymbol{U}^\top \boldsymbol{Q}^\top \boldsymbol{Q} \boldsymbol{U}$ which from the previous calculations is equal to $\sigma^2\boldsymbol{U}^\top \boldsymbol{Q} \boldsymbol{U} = \sigma^2 \Lambda$,  where $\Lambda$ is the diagonal matrix of the eigenvalues of $\boldsymbol{Q}$, which can only be 0 or 1 due to $\boldsymbol{Q}$ being a projection matrix. Then, we have that $p$ of the entries of $\hat{W}'$ are equal to zero, and so $\|\hat{W}'\|_2^2/\sigma^2$ is the sum of the squares $n-p$ independent Gaussian random variables each with mean zero and variance one, which is exactly a chi-squared distribution with $n-p$ degrees of freedom. Finally, notice that $\|\hat{W} \|_2^2 = \|\boldsymbol{U}^\top\hat{W} \| = \|\hat{W}'\|_2^2$.
\end{proof}


\bibliographystyle{IEEEtran}
 \bibliography{refs_shaska_updated}

\end{document}